%% file: main.tex
\documentclass[sigconf]{acmart}

\usepackage{amsmath,amssymb,amsfonts,amsthm}
\usepackage{algorithmic}
\usepackage[shortlabels]{enumitem}
\usepackage{graphicx}
\usepackage{textcomp}
\usepackage{xcolor}
\usepackage{url}
\usepackage{xspace}
\usepackage{cleveref}
\usepackage{pifont}
\usepackage{booktabs}
\usepackage{wrapfig}
\usepackage[font=footnotesize]{caption}
\usepackage{listings}

\def\cameraReady{} 

\input{macros}
\graphicspath{{figures/}}

\AtBeginDocument{
  \providecommand\BibTeX{{
    \normalfont B\kern-0.5em{\scshape i\kern-0.25em b}\kern-0.8em\TeX}}}

\begin{document}

\title{\sysname: High-Performance Byzantine Fault Tolerant Settlement}

\ifdefined\cameraReady
\author{Mathieu Baudet}
\authornote{Alphabetical order.}
\email{mathieubaudet@fb.com}
\affiliation{
    \institution{Facebook Novi}
}
\author{George Danezis}
\email{gdanezis@fb.com}
\affiliation{
    \institution{Facebook Novi}
}
\author{Alberto Sonnino}
\email{asonnino@fb.com}
\affiliation{
    \institution{Facebook Novi}
}
\else
  \author{}
\fi

\input{sections/abstract}

\keywords{distributed system, bft, settlement system, consistent broadcast}

\maketitle

\input{sections/introduction.tex}
\input{sections/background.tex}
\input{sections/overview.tex}

\input{sections/design.tex}
\input{sections/security.tex}
\input{sections/implementation.tex}
\input{sections/evaluation.tex}

\input{sections/discussion.tex}
\input{sections/related.tex}

\input{sections/conclusion.tex}

\ifdefined\cameraReady
\begin{acks}
This work is funded by Novi, a Facebook subsidiary. The authors would like to thank Dahlia Malkhi for her constant feedback on this project, and Kostas Chalkias for feedback on late manuscript. We also thank the Novi Research and Engineering teams for valuable feedback.
\end{acks}
\fi

\bibliographystyle{ACM-Reference-Format}
\bibliography{references}

\appendix
\input{sections/proofs.tex}
\input{sections/code-listings.tex}
\input{sections/additional-figures.tex}
\end{document}

%% file: macros.tex
\newcommand\bftlong{Byzantine Fault-Tolerant\xspace}

\newcommand\libra{Libra\xspace}

\newcommand\sysname{FastPay\xspace}
\newcommand\sysmain{Primary\xspace}

\newcommand{\keyword}[1]{\normalfont \texttt{#1}}
\newcommand{\accounts}{\keyword{accounts}}
\newcommand{\accountbalance}{\keyword{balance}}
\newcommand{\totalbalance}{\keyword{total\_balance}}
\newcommand{\val}{\keyword{value}}
\newcommand{\sequencenumber}{\keyword{sequence}}
\newcommand{\lasttransactionindex}{\keyword{last\_transaction}}
\newcommand{\transactionindex}{\keyword{transaction\_index}}
\newcommand{\nextsequencenumber}{\keyword{next\_sequence}}
\newcommand{\redeemlog}{\keyword{redeemed}}
\newcommand{\synchronizationlog}{\keyword{synchronized}}
\newcommand{\funding}{\keyword{funding}}

\newcommand{\confirmedlog}{\keyword{confirmed}}
\newcommand{\pendingconfirmation}{\keyword{pending}}
\newcommand{\amount}{\keyword{amount}}
\newcommand{\sender}{\keyword{sender}}
\newcommand{\recipient}{\keyword{recipient}}
\newcommand{\receivedlist}{\keyword{received}}
\newcommand{\transactions}{\keyword{transactions}}
\newcommand{\sysnameaddresses}{\keyword{FastPay}}
\newcommand{\sysmainaddresses}{\keyword{Primary}}
\newcommand{\transfer}{O}
\newcommand{\cert}{C}
\newcommand{\sync}{S}

\newcommand{\authority}{\alpha}

\newcommand{\para}[1]{\vspace{2mm}\noindent\textbf{#1}.\xspace}

\newcommand{\etal}{\textit{et al.}\@\xspace}
\newcommand{\eg}{\textit{e.g.}\@\xspace}
\newcommand{\ie}{\textit{i.e.}\@\xspace}

\def\first{({i})\xspace}
\def\second{({ii})\xspace}
\def\third{({iii})\xspace}
\def\fourth{({iv})\xspace}
\def\fifth{({v})\xspace}
\def\sixth{({vi})\xspace}

\definecolor{verylightgray}{gray}{0.9}

\lstdefinelanguage{alg}{
morekeywords={fn, let, mut, ensure, match, bail, if, None, Some, return, Result, Ok},
morecomment=[l]///,
}

%% file: sections/abstract.tex
\begin{abstract}
\sysname allows a set of distributed authorities, some of which are Byzantine, to maintain a high-integrity and availability settlement system for pre-funded payments. It can be used to settle payments in a native unit of value (crypto-currency), or as a financial side-infrastructure to support retail payments in fiat currencies. \sysname is based on Byzantine Consistent Broadcast as its core primitive, foregoing the expenses of full atomic commit channels (consensus). The resulting system has low-latency for both confirmation and payment finality. Remarkably, each authority can be sharded across many machines to allow unbounded horizontal scalability. Our experiments demonstrate intra-continental confirmation latency of less than 100ms, making \sysname applicable to point of sale payments. In laboratory environments, we achieve 
over 80,000 transactions per second with 20 authorities---surpassing the requirements of current retail card payment networks, while significantly increasing their robustness.
\end{abstract}

%% file: sections/introduction.tex
\section{Introduction} \label{sec:introduction}
Real-time gross settlement systems (RTGS)~\cite{rtgs-2} constitute the most common approach to financial payments in closed banking networks, that is, between reputable institutions. In contrast, blockchain platforms have proposed a radically different paradigm, allowing account holders to interact directly with an online, yet highly secure, distributed ledger. Blockchain approaches aim to enable new use cases such as personal e-wallets or private transactions, and generally provide ecosystems more favorable to users. However, until now, such open, distributed settlement solutions have come at a high performance cost and questionable scalability compared to traditional, closed RTGS systems.

\sysname is a Byzantine Fault Tolerant (BFT) real-time gross settlement (RTGS) system. It enables authorities to jointly maintain account balances and settle pre-funded retail payments between accounts. It supports extremely low-latency confirmation (sub-second) of eventual transaction finality, appropriate for physical point-of-sale payments. It also provides extremely high capacity, comparable with peak retail card network volumes, while ensuring gross settlement in real-time. \sysname eliminates counterparty and credit risks of net settlement and removes the need for intermediate banks, and complex financial contracts between them, to absorb these risks. \sysname can accommodate arbitrary capacities through efficient sharding architectures at each authority. Unlike any traditional RTGS, and more like permissioned blockchains, \sysname can tolerate up to $f$ Byzantine failures out of a total of $3f+1$ authorities, and retain both safety, liveness, and high-performance.


\sysname can be deployed in a number of settings. First, it may be used as a settlement layer for a native token and crypto-currency, in a standalone fashion. Second, it may be deployed as a side-chain of another crypto-currency, or as a high performance settlement layer on the side of an established RTGS to settle fiat retail payments. In this paper we present this second functionality in detail, since it exercises all features of the system, both payments between \sysname accounts, as well as payments into and out of the system. 

\para{Contributions} We make the following contributions:
\begin{itemize}
\setlength\itemsep{0em}
    \item The \sysname design is novel in that if forgoes full consensus; it leverages the semantics of payments to minimize shared state between accounts and to increase the concurrency of asynchronous operations; and supports sharded authorities.
    \item We provide proofs of safety and liveness in a Byzantine and fully asynchronous network setting. We show that \sysname keeps all its properties despite the lack of total ordering, or asynchrony of updates to recipient accounts. 
    \item We experimentally demonstrate comparatively very high throughput and low latency, as well as the robustness of the system under conditions of extremely high concurrency and load. We show that performance is maintained even when some (Byzantine) authorities fail.
\end{itemize}

\para{Outline} This paper is organized as follows:
 \Cref{sec:background} introduces real-time gross settlement systems, and permissioned blockchains.
 \Cref{sec:overview} introduces the entities within \sysname, their interactions, and the security properties and threat model.
 \Cref{sec:design} details the design of \sysname both as a standalone system, and operated in conjunction with a \sysmain. 
 \Cref{sec:security} discusses safety and liveness.
 \Cref{sec:implementation} briefly describes the implementation of the \sysname prototype.
 \Cref{sec:evaluation} provides a full performance evaluation of \sysname as we modulate its security parameters and load. 
 \Cref{sec:discussion} discusses key open issues such as privacy, governance mechanisms and economics of the platform.
 \Cref{sec:related} covers the related work, both in terms of traditional financial systems and crypto-currencies.
 \Cref{sec:conclusion} concludes.

%% file: sections/background.tex
\section{Background} \label{sec:background}

Real-time gross settlement systems (RTGS)~\cite{rtgs-2} are the backbone of modern financial systems. Commercial banks use them to maintain an account with central banks and settle large value payments. 

RTGS systems are limited in their capacity\footnote{For example, the relatively recent European Central Bank TARGET2 system has a maximum capacity of 500 transactions per second~\cite{target2}.}, making them unsuitable for settling low-value high-volume retail payments directly. Such retail payments are deferred: banks exchange information bilaterally about pending payments (often through SWIFT~\cite{swift, swift-iso}), they aggregate and net payments, and only settle net balances through an RTGS, often daily. The often quoted volume figure of around 80,000 transactions per second for retail card networks~\cite{mastercard-performance,visa-performance} represents the rate at which `promises' for payments are exchanged between banks, and not settled payments.
Traditional RTGS systems are implemented as monolithic centralized services operated by a single authority, and must employ a number of technical and organizational internal controls to ensure they are reliable (through a primary-replica architecture with manual switch over) and correct---namely ensuring availability and integrity. Traditionally only regulated entities have accounts in those systems. This result in a Balkanized global financial system where financial institutions connect to multiple RTGS, directly or indirectly through corresponding banks, to execute international payments.

Blockchain-based technologies, starting with Bitcoin~\cite{bitcoin} in 2009, provide more open settlement systems often combined with their own digital tokens to represent value. Permissionless blockchains have been criticized for their low performance~\cite{croman2016scaling} in terms of capacity and finality times. However, a comparison with established settlement systems leads to a more nuanced assessment. Currently, Ethereum~\cite{wood2014ethereum} can process 15 transactions per second. The actual average daily load on the EU ECB TARGET2 system is about 10 transactions per second~\cite{target2} (in 2018) which is a comparable figure (and lower than the peak advertised capacity of 500 transaction per second). However, it falls very short of the advertised transaction rate of 80,000 transaction per second peak for retail payment networks---even though this figure does not represents settled transactions. The stated ambitions of permissionless projects is to be able to settle transactions at this rate on an open and permissionless network, which remains an open research challenge~\cite{vukolic2015quest}. 

Permissioned blockchains~\cite{permissioned-blockchains, sok-consensus} provide a degree of decentra\-lization---allowing multiple authorities to jointly operate a ledger---at the cost of some off-chain governance to control who maintains the blockchain. The most prominent of such proposals is the Libra network~\cite{libra}, 
developed by the Libra Association. Other technical efforts include Hyperledger~\cite{hyperledger}, Corda~\cite{corda} and Tendermint~\cite{tendermint}.  These systems are based on traditional notions of Byzantine Fault Tolerant state machine replication (or sometimes consensus with crash-failures only), which presupposes an atomic commit channel (often referred as `consensus') that sequences all transactions. Such architectures allow for higher capacities than Bitcoin and Ethereum. LibraBFT, for example, aims for 1,000 transactions per second at peak capacity~\cite{libraTPS-1, libraTPS-2}; this exceeds many RTGS systems but is still below the peak volumes for retail payment systems. Regarding transaction finality, a~latency of multiple seconds is competitive with RTGS systems but is not suitable for retail payment at physical points of sale. 

%% file: sections/overview.tex
\section{Overview} \label{sec:overview}
To illustrate its full capabilities, we describe \sysname as a side chain of a primary RTGS holding the primary records of accounts. We call such a primary ledger \emph{the \sysmain} for short, and its accounts \emph{the \sysmain accounts}. The \sysmain can be instantiated in two ways: \first as a programmable blockchain, through smart contracts, like Ethereum~\cite{wood2014ethereum} or \libra~\cite{libra}. The \sysmain can also be instantiated \second as a traditional monolithic RTGS operated by a central bank. In this case the components interfacing with \sysname are implemented as database transactions within the \sysmain. In both cases \sysname acts as a side infrastructure to enable pre-funded retail payments. \sysname can also operate with a native asset, without a primary ledger. In this case sub-protocols involving the \sysmain are superfluous, since all value is held within \sysname accounts and never transferred out or into the system. 

\subsection{Participants} \label{sec:participants}
\sysname involves two types of participants: \first authorities, and \second account owners (\emph{users}, for short). All participants generate a key pair consisting of a private signature key and the corresponding public verification key.
As a side-chain, \sysname requires a smart contract on the main blockchain, or a software component on an RTGS system that can authorize payments based on the signatures of a threshold of authorities from a \emph{committee} with fixed membership.

%
%

By definition, an \emph{honest} authority always follows the \sysname protocol, while a \emph{faulty} (or \emph{Byzantine}) one may deviate arbitrarily.
We present the \sysname protocol for $3f+1$ equally-trusted authorities, assuming a fixed (but unknown) subset of at most $f$ Byzantine authorities. In this setting, a \emph{quorum} is defined as any subset of $2f+1$ authorities. (As for many BFT protocols, our proofs only use the classical properties of quorums thus apply to all Byzantine quorum systems~\cite{malkhi1998byzantine}.)

When a protocol message is signed by a quorum of authorities, it is said to be \emph{certified}: we call such a jointly signed message a \emph{certificate}.

\subsection{Accounts and Actions} \label{sec:account-actions}

A \sysname \emph{account} is identified by its \emph{address}, which we instantiate as the cryptographic hash of its public verification key. The state of a \sysname account is affected by four main high-level actions:
\begin{enumerate}
\setlength\itemsep{0em}
    \item Receiving funds from a \sysmain account.
    \item Transferring funds to a \sysmain account.
    \item Receiving funds from a \sysname account.
    \item Transferring funds to a \sysname account.
\end{enumerate}
%
\sysname also supports two read-only actions that are necessary to ensure liveness despite faults: reading the state of an account at a \sysname authority, and obtaining a certificate for any action executed by an authority.

\subsection{Protocol Messages} \label{sec:protocol-messages}
The \sysname protocol consists of \emph{transactions} on the \sysmain, denoted with letter $T$, and network requests that users send to \sysname authorities, which we call \emph{orders}, and denote with letter $O$. Users are responsible for broadcasting their orders to authorities and for processing the corresponding responses. The authorities are passive and \emph{do not communicate directly with each other}.

\para{Transfer orders} All transfers initiated by a \sysname account start with a \emph{transfer order} $O$ including the following fields:
\begin{itemize}
\setlength\itemsep{0em}
    \item The sender's \sysname address, written $\sender(O)$.
    \item The recipient, either a \sysname or a \sysmain address, written $\recipient(O)$.
    \item A non-negative amount to transfer, written $\amount(O)$.
    \item A sequence number $\sequencenumber(O)$.
    \item Optional user-provided data.
    \item A signature by the sender over the above data.
\end{itemize}

Authorities respond to valid transfer orders by signing them (see next section for validity checks). A quorum of such signatures is meant to be aggregated into a \emph{transfer certificate}, noted $C$.

\para{Notations} We write $O = \val(C)$ for the original transfer order $O$ certified by $C$. For simplicity, we omit the operator $\val$ when the meaning is clear, e.g. $\sender(C) = \sender(\val(C))$.
\sysname addresses are denoted with letters $x$ and $y$. We use $\alpha$ for authorities and by extension for the shards of authorities.

\subsection{Security Properties and Threat Model} \label{sec:properties}
\sysname guarantees the following security properties:
\begin{itemize}
\setlength\itemsep{0em}
    \item \textbf{Safety:} No units of value are ever created or destroyed; they are only transferred between accounts.
    \item \textbf{Authenticity:} Only the owner of an account may transfer value out of the account.
    \item \textbf{Availability:} Correct users can always transfer funds from their account.
    \item \textbf{Redeemability:} A transfer to \sysname or \sysmain is guaranteed to eventually succeed whenever a valid transfer certificate has already been produced.
    \item \textbf{Public Auditability:} There is sufficient public cryptographic evidence for the state of \sysname to be audited for correctness by any party. 
    \item \textbf{Worst-case Efficiency:} Byzantine authorities (or users) cannot significantly delay operations from correct users.
\end{itemize}

The above properties are maintained under a number of security assumptions: \first~there are at most $f$ Byzantine authorities out of $3f+1$ total authorities.
\second~The network is fully asynchronous, and the adversary may arbitrarily delay and reorder messages~\cite{dwork1988consensus}. However, messages are eventually delivered. \third~Users may behave arbitrarily but availability only holds for \emph{correct users} (defined in \Cref{sec:clients}). \fourth~The \sysmain provides safety and liveness (when \sysname is used in conjunction with it).
We further discuss the security properties of \sysname in \Cref{sec:security}.


%% file: sections/design.tex
\section{The \sysname Protocol} \label{sec:design}
\sysname authorities hold and persist the following information.

\para{Authorities}
The state of an authority $\alpha$ consists of the following information:
\begin{itemize}
\setlength\itemsep{0em}
    \item The authority name, signature and verification keys.
    \item The committee, represented as a set of authorities and their verification keys.
    \item A map $\accounts(\alpha)$ tracking the current account state of each \sysname address $x$ in use (see below).
    \item An integer value, noted $\lasttransactionindex(\alpha)$, referring to the last transaction that paid funds into the \sysmain. This is used by authorities to synchronize \sysname accounts with funds from the \sysmain (see Section~\ref{sec:libra-to-fastpay}).
\end{itemize}

\para{\sysname accounts}
The state of a \sysname account $x$ within the authority $\alpha$ consists of the following:
\begin{itemize}
\setlength\itemsep{0em}
    \item The public verification key of $x$, used to authenticate spending actions.
    \item An integer value representing the balance of payment, written $\accountbalance^x(\alpha)$.
    \item An integer value, written $\nextsequencenumber^x(\alpha)$, tracking the expected sequence number for the next spending action to be created. This value starts at $0$. 
    \item A field $\pendingconfirmation^x(\alpha)$ tracking the last transfer order $O$ signed by $x$ such that the authority $\alpha$ considers $O$ as \emph{pending confirmation}, if any; and absent otherwise.
    \item A list of certificates, written $\confirmedlog^x(\alpha)$, tracking all the transfer certificates $C$ that have been \emph{confirmed} by $\alpha$ and such that $\sender(C) = x$. One such certificate is available for each sequence number $n$ ($0 \leq n < \nextsequencenumber^x(\alpha)$).
    \item A list of \emph{synchronization} orders, written $\synchronizationlog^x(\alpha)$, tracking transferred funds from the \sysmain to account $x$. (See Section~\ref{sec:libra-to-fastpay}.)
\end{itemize}
We also define $\receivedlist^x(\alpha)$ as the list of confirmed certificates for transfers received by $x$. Formally, $\receivedlist^x(\alpha) = \{C \text{ s.t. } \exists y.\; C \in \confirmedlog^y(\alpha) \text{ and } \recipient(C) = x \}$.
 
 We assume arbitrary size integers. Although \sysname does not let users overspend, (temporary) negative balances for account states are allowed for technical reasons discussed in \Cref{sec:security}. When present, a pending (signed) transfer order $O = \pendingconfirmation^x(\alpha)$ effectively locks the sequence number of the account $x$ and prevents $\alpha$ from accepting new transfer orders from $x$ until \emph{confirmation}, that is, until a valid transfer certificate $C$ such that $\val(C) = O$ is received. This mechanism can be seen as the `Signed Echo Broadcast' implementation of a Byzantine consistent broadcast on the label (account, next sequence number)~\cite{cachinBook}.

\para{Storage considerations} The information contained in the lists of certificates $\confirmedlog^x(\alpha)$ and $\receivedlist^x(\alpha)$ and in the synchronization orders $\synchronizationlog^x(\alpha)$ is self-authenticated---being respectively signed by a quorum of authorities and by the \sysmain. Remarkably, this means that authorities may safely outsource these lists to an external high-availability data store. Therefore, \sysname authorities only require a constant amount of local storage per account, rather than a linear amount in the number of transactions.

\subsection{Transferring Funds within \sysname} \label{sec:fastpay-to-*}

\sysname operates by implementing a Byzantine consistent broadcast channel per account, specifically using a `Signed Echo Broadcast' variant (Algorithm 3.17 in~\cite{cachinBook}). It operates in two phases and all messages are relayed by the initiating user. Consistent Broadcast ensures \emph{Validity}, \emph{No duplication}, \emph{Integrity}, and \emph{Consistency}. It always terminates when initiated by a correct user. However, if a \sysname user equivocates, current operations may fail, and the funds present on the users account may become inaccessible.

\begin{figure}[t]
\includegraphics[width=.48\textwidth]{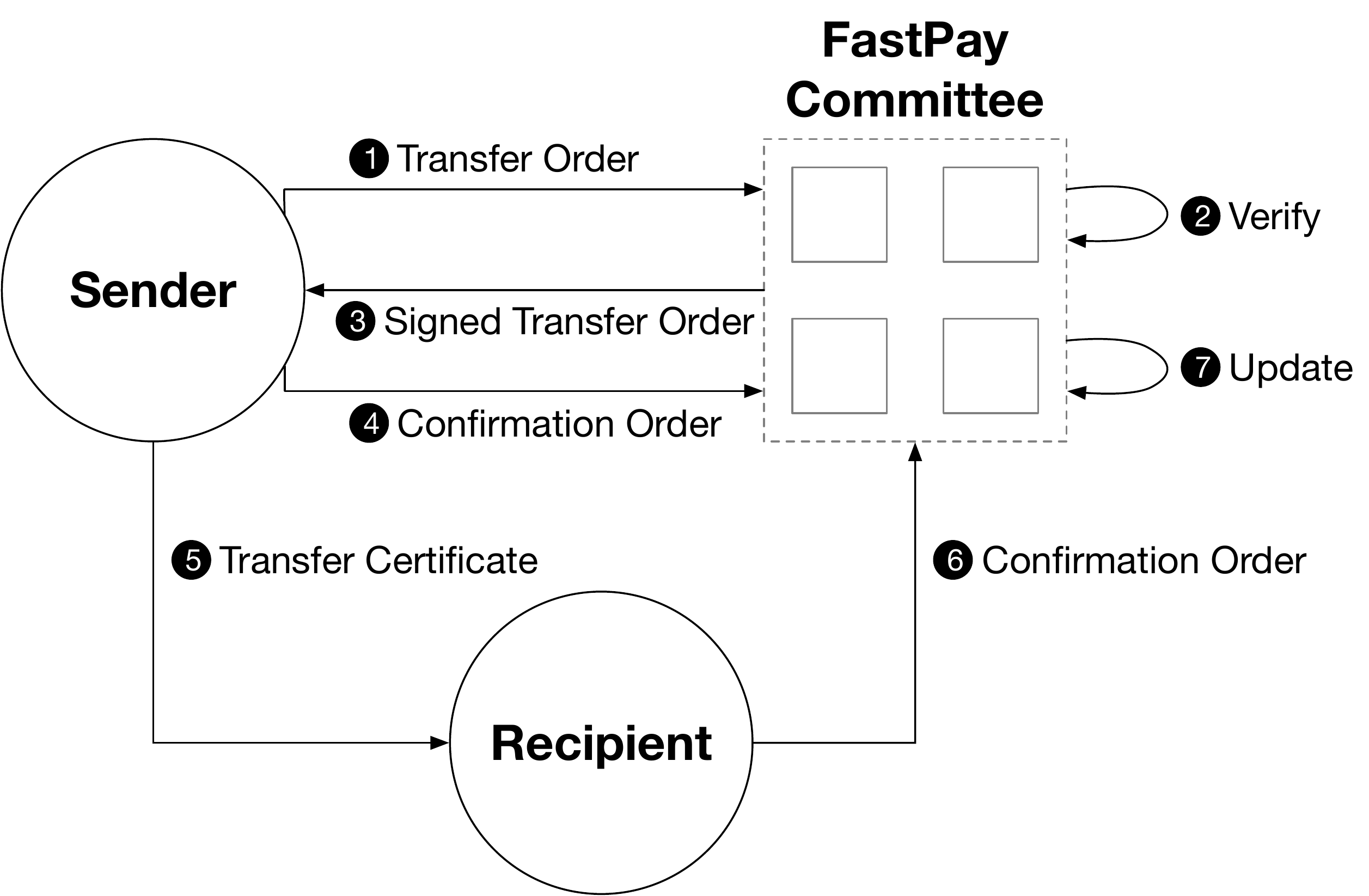}
\caption{Transfer of funds from \sysname to \sysname.}
\label{fig:fastpay-to-fastpay}
\end{figure}

\para{Transferring funds}
\Cref{fig:fastpay-to-fastpay} illustrates a transfer of funds within \sysname.
To transfers funds to another \sysname account, the sender creates a \sysname transfer order ($\transfer$) with the next sequence number in their account, and signs it. They then send the \sysname transfer order to all authorities. Each authority checks (\ding{202}) \first that the signature is valid, \second that no previous transfer is pending (or is for the same transfer), \third that the amount is positive, \fourth that the sequence number matches the expected next one ($\sequencenumber(O) = \nextsequencenumber^x(\authority)$), and \fourth that the balance ($\accountbalance^x(\alpha)$) is sufficient (\ding{203}). Then, it records the new transfer as pending and sends back a signature on the transfer order (\ding{204}) which is also stored. The authority algorithm to handle transfer orders, corresponding to step \ding{203}, is presented in \Cref{fig:code}.

The user collects the signatures from a quorum of authorities, and uses them along the \sysname transfer order to form a transfer certificate.
The sender provides this transfer certificate to the recipient as proof that the payment will proceed (\ding{206}).
To conclude the transaction, the sender (\ding{205}) or the recipient (\ding{207}) broadcast the transfer certificate ($\cert$) to the authorities (called \emph{confirmation order})\footnote{Aggregating signed transfer orders into a transfer certificate does not requires knowledge of any secret; therefore, anyone (and not only the sender or the recipient) can broadcast the transfer certificate to the authorities to conclude the transaction.}.


Upon reception of a confirmation order for the current sequence number, each authority $\alpha$ (\ding{208}) \first checks that a quorum of signatures was reached, \second decreases the balance of the sender, \third increments the sequence number ($\nextsequencenumber^x(\alpha) + 1$) to ensure `deliver once' semantics, and \fourth sets the pending order to None ($\pendingconfirmation^x(\alpha)=\text{None}$). Each authority $\alpha$ also \fifth adds the certificate to the list $\confirmedlog^x(\alpha)$, and \sixth increases the balance of the recipient account asynchronously (\ie without sequencing this write in relation to any specific payments from this account across authorities). The authority algorithm to handle confirmation orders (as in step \ding{208}) is presented in \Cref{fig:code}.

In \Cref{sec:security} and \Cref{sec:proofs}, we show that the \sysname protocol is safe thanks to the semantics of payments into an account and their commutative properties. \sysname is a significant simplification and deviation from an orthodox application of Guerraroui~\etal~\cite{consensus-number}, where accounts are single-writer objects and all write actions are mediated by the account owner. \sysname allows payments to be executed after a single consistent broadcast, rather than requiring recipients to sequence payments into their accounts separately. This reduces both latency and the state necessary to prevent replays.

\begin{figure}
\begin{lstlisting}[mathescape=true]
fn handle_transfer_order($\authority$, $\transfer$) -> Result {
    /// Check shard and signature.
    ensure!($\authority$.in_shard($\sender(\transfer)$));
    ensure!($\transfer$.has_valid_signature());

    /// Obtain sender account.
    match $\accounts(\authority)$.get($\sender(\transfer)$) {
        None => bail!(),
        Some(account) => {
            /// Check if the same order is already pending.
            if let Some(pending) = account.pending {
                ensure!(pending.transfer == $\transfer$);
                return Ok();
            }
            ensure!(account.next_sequence == $\sequencenumber(\transfer)$);
            ensure!(account.balance >= $\amount(\transfer)$);
            /// Sign and store new transfer.
            account.pending = Some($\authority$.sign($\transfer$));
            return Ok();
}   }   }

fn handle_confirmation_order($\authority$, $\cert$)
    -> Result<Option<CrossShardUpdate>> {
    /// Check shard and certificate.
    ensure!($\authority$.in_shard($\sender(\cert)$));
    ensure!($\cert$.is_valid($\authority$.committee));
    let $\transfer$ = $\val(\cert)$;

    /// Obtain sender account.
    let sender_account =
        $\accounts(\authority)$.get($\sender(\transfer)$)
        .or_insert(AccountState::new());

    /// Ignore old certificates.
    if sender_account.next_sequence > $\sequencenumber(\transfer)$ {
        return Ok(None);
    }

    /// Check sequence number and balance.
    ensure!(sender_account.next_sequence == $\sequencenumber(\transfer)$);
    ensure!(sender_account.balance >= $\amount(\transfer)$);

    /// Update sender account.
    sender_account.balance -= $\amount(\transfer)$;
    sender_account.next_sequence += 1;
    sender_account.pending = None;
    sender_account.confirmed.push($\cert$);

    /// Update recipient locally or cross-shard.
    let recipient = match $\recipient(\transfer)$ {
        Address::FastPay(recipient) => recipient,
        Address::Primary(_) => { return Ok(None) }
    };

    /// Same shard: read and update the recipient.
    if $\authority$.in_shard(recipient) {
        let recipient_account = $\accounts(\authority)$.get(recipient)
            .or_insert(AccountState::new());
        recipient_account.balance += $\amount(\transfer)$;
        return Ok(None);
    }

    /// Other shard: request a cross-shard update.
    let update = CrossShardUpdate {
        shard_id: $\authority$.which_shard(recipient),
        transfer_certificate: $\cert$,
    };
    Ok(Some(update))
}
\end{lstlisting}
\caption{Authority algorithms for handling transfer and confirmation orders. (The cross-shard update logic is presented in \Cref{sec:code-listings}.)}
\label{fig:code}
\end{figure}

\para{Payment finality} 
Once a transfer certificate \emph{could} be formed, namely $2f+1$ authorities signed a transfer order, no other order can be processed for an account until the corresponding confirmation order is submitted. Technically, the payment is final: it cannot be canceled, and will proceed eventually. As a result, showing a transfer certificate to a recipient convinces them that the payment will proceed. We call the showing of a transfer certificate to a recipient a \emph{confirmation}, and then subsequently submitting the confirmation order, to move funds, \emph{settlement}. \emph{Confirmation} requires only a single round trip to a quorum of authorities resulting in very low-latency (see \Cref{sec:evaluation}), and giving the system its name.

\para{Proxies, gateways and crash recovery} The protocols as presented involve the sender being on-line and mediating all communications.
However, the only action that the sender \emph{must} perform personally is forming a transfer order, requiring their signature. All subsequent operations, including sending the transfer order to the authorities, forming a certificate, and submitting a confirmation order can be securely off-loaded to a proxy trusted only for liveness. Alternatively, a transfer order may be given to a merchant (or payment gateway) that drives the protocol to conclusion. In fact, any party in possession of a signed transfer order may attempt to make a payment progress concurrently. And as long as the sender is correct the protocol will conclude (and if not may only lock the account of the faulty sender).
This provides significant deployment and implementation flexibility. A sender client may be implemented in hardware (in a NFC smart card) that only signs transfer orders. These are then provided to a gateway that drives the rest of the protocol. Once the transfer order is signed and handed over to the gateway, the sender may go off-line or crash. Authorities may also attempt to complete the protocol upon receiving a valid transfer order.
Finally, the protocol recovers from user crash failures: anyone may request a transfer order that is partially confirmed from any authority, proceed to form a certificate, and submit a confirmation order to complete the protocol.

\subsection{Sharding authorities} \label{sec:sharding}
\sysname requires minimal state sharing between accounts, and allows for a very efficient sharding at each authority by account. The consistent broadcast channel is executed on a per-account basis. Therefore, the protocol does not require any state sharing between accounts (and shards) up to the point where a valid confirmation order has to be settled to transfer funds between \sysname accounts.
On settlement, the sender account is decremented and the funds are deposited into the account of the recipient, requiring interaction between at most two shards (second algorithm of \Cref{fig:code}).

Paying into an account can be performed asynchronously, and is an operation that cannot fail (if the account does not exist it is created on the spot).  Therefore, the shard managing the recipient account only needs to be notified of the confirmed payment through a reliable, deliver once, authenticated, point to point channel (that can be implemented using a message authentication code, inter-shard sequence number, re-transmission, and acknowledgments) from the sender shard. This is a greatly simplified variant of a two-phase commit protocol coordinated by the sender shard (for details see the Presume Nothing and Last Agent Commit optimizations~\cite{DBLP:journals/dpd/SamarasBCM95,DBLP:conf/vldb/LampsonL93}). Modifying the validity condition of the consistent broadcast to ensure the recipient account exists (or any other precondition on the recipient account) would require a full two-phase commit before an authority signs a transfer order, and can be implemented while still allowing for (slightly less) efficient sharding.

The algorithms in \cref{fig:code} implement sharding. An authority shard checks that the transfer order ($\transfer$) or certificate ($\cert$) is to be handled by a specific shard and otherwise rejects it without mutating its state. Handling confirmation orders depends on whether a recipient account is on the same shard. If so, the recipient account is updated locally. Otherwise, a \emph{cross shard message} is created for the recipient shard to update the account (see code in the Appendix for this operation).
The ability to shard each authority has profound implications: increasing the number of shards at each authority increases the theoretical throughput linearly, while latency remains constant. Our experimental evaluation confirms this experimentally (see \Cref{sec:evaluation}).

\subsection{Interfacing with the \sysmain} \label{sec:onchain-state}
We describe the protocols required to couple \sysname with the \sysmain, namely transferring funds from the \sysmain to a \sysname account, and conversely from a \sysname to a \sysmain account.
We refer throughout to the logic on the \sysmain as a \emph{smart contract}, and the primary store of information as the \emph{blockchain}. 
A traditional RTGS would record this state and manage it in conventional ways using databases and stored procedures, rather than a blockchain and smart contracts. We write $\sigma$ for the state of the `blockchain' at a given time, and $\transactions(\sigma)$ for the set of \sysname transactions $T$ already processed by the blockchain.

\para{Smart contract}
The smart contract mediating interactions with the \sysmain requires the following data to be persisted in the blockchain:
\begin{itemize}
\setlength\itemsep{0em}
    \item The \sysname committee composition: a set of authority names and their verification keys.
    \item A map of accounts where each \sysname address is mapped to its current \sysmain state (see below).
    \item The total balance of funds in the smart contract, written $\totalbalance(\sigma)$.
    \item The transaction index of the last transaction that added funds to the smart contract, written $\lasttransactionindex(\sigma)$.
\end{itemize}

\para{Accounts}
The \sysmain state of a \sysname account $x$ consists of the set of sequence numbers of transfers already executed from this account to the \sysmain. This set is called the \emph{redeem log} of $x$ and written $\redeemlog^x(\sigma)$.

\para{Adding funds from the \sysmain to \sysname} \label{sec:libra-to-fastpay}
\Cref{fig:libra-to-fastpay} shows a transfer of funds from the \sysmain to \sysname. 
The owner of the \sysname account (or anyone else) starts by sending a payment to the \sysname smart contract using a \sysmain transaction\ (\ding{202}). This transaction is called a \emph{funding transaction}, and includes the recipient \sysname address for the funds and the amount of value to transfer.

\begin{figure}[t]
\includegraphics[width=.48\textwidth]{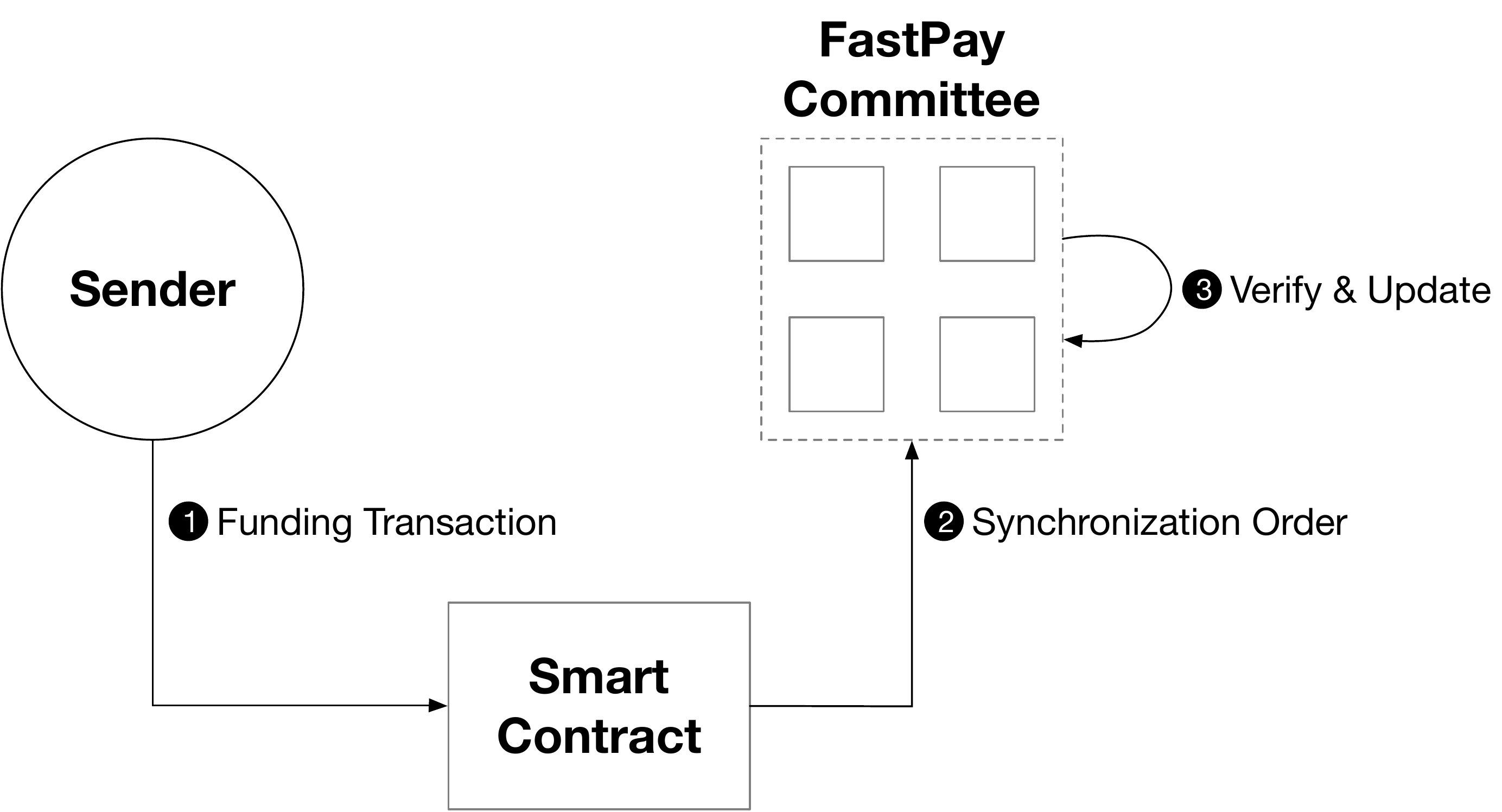}
\caption{Transfer of funds from the \sysmain to \sysname.}
\label{fig:libra-to-fastpay}
\end{figure}

When the \sysmain transaction is executed, the \sysname smart contract generates a \emph{\sysmain event} that instructs authorities of a change in the state of the \sysname smart contract. We assume that each authority runs a full \sysmain client to authenticate such events. For simplicity, we  model such an event as a \emph{(\sysmain) synchronization order} (\ding{203}). The smart contract ensures this event and the synchronization order contain a unique, always increasing, sequential transaction index.

When receiving a synchronization order, each authority \first checks that the transaction index follows the previously recorded one, \second increments the last transaction index in their global state, \third creates a new \sysname account if needed, and \fourth increases the account balance of the target account by the amount of value specified (\ding{204}). \Cref{sec:code-listings} presents the authority algorithm for handling funding transactions.

\para{Transferring funds from \sysname to the \sysmain}
\Cref{fig:fastpay-to-libra} shows a transfer of funds from \sysname to the \sysmain. The \sysname sender signs a \emph{\sysmain transfer order} using their account key and broadcasts it to the authorities (\ding{202}). This is simply a transfer order with a \sysmain address as the recipient.

\begin{figure}[t]
\includegraphics[width=.48\textwidth]{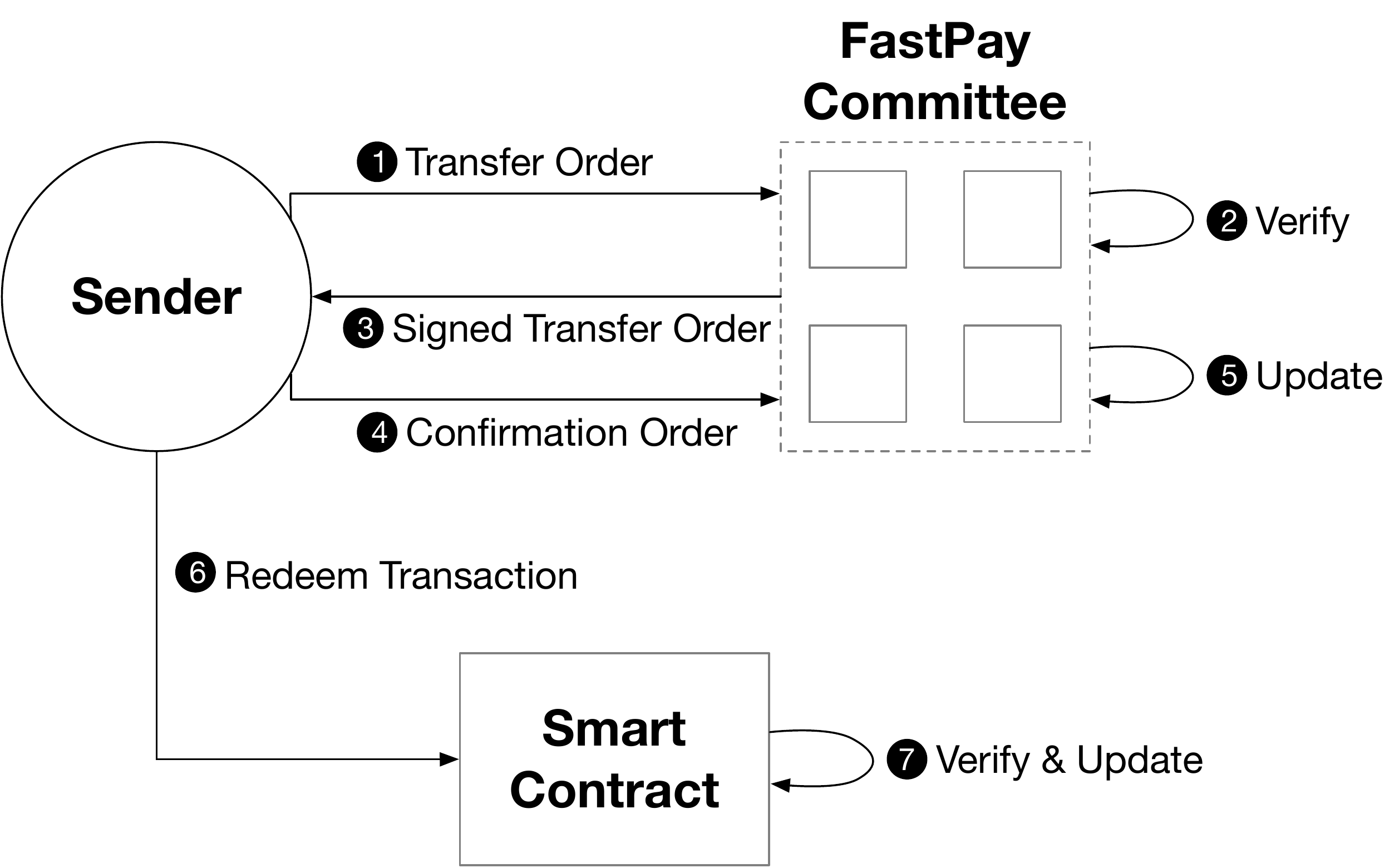}
\caption{Transfer of funds from \sysname to the \sysmain.}
\label{fig:fastpay-to-libra}
\end{figure}

Once a quorum of signatures is reached (\ding{203} and \ding{204}), the sender creates a certified (\sysmain) transfer order, also called a \emph{transfer certificate} for short.
The sender broadcasts this certificate to the authorities to confirm the transaction (\ding{205}) and unlock future spending from this account. When an authority receives a confirmation order containing a certificate of transfer (\ding{206}), it must check \first that a quorum of signatures was reached, and \second that the account sequence number matches the expected one; they \third then set the pending order to None, \fourth increment the sequence number, and \fifth decrease the account balance.

Finally, the recipient of the transfer should send a redeem transaction to the \sysname smart contract on the \sysmain blockchain (\ding{207}). When the \sysname smart contract receives a valid redeem transaction (\ding{208}), it must \first check that the sequence number is not in the \sysmain redeem log of the sender, to prevent reuse; \second update this redeem log; \fourth transfer the amount of value specified from the smart contract into the recipient's \sysmain account.

\subsection{State Recovery and Auditing} \label{sec:audits}

For every account $x$, each authority $\alpha$ must make available the pending order $\pendingconfirmation^x(\alpha)$, the sequence number $\nextsequencenumber^x(\alpha)$, the synchronization orders $\synchronizationlog^x(\alpha)$, and the certificates confirmed so far, indexed by senders (\ie $\confirmedlog^x(\alpha)$) and receivers ($\receivedlist^x(\alpha)$).
Sharing these data fulfills two important roles: \first this lets anyone read the state of any incomplete transfer and drive the protocol all the way to settlement; \second it enables auditing authority states and detecting Byzantine faults (\eg incorrect balance checks).

\subsection{Correct Users and Client Implementation} \label{sec:clients}

A \emph{correct user} owning a \sysname account $x$ follows the correctness rules below:
\begin{enumerate}
\setlength\itemsep{0em}
\item The user sets the sequence number of a new transfer order $O$ to be the next expected integer after the previous transfer (starting with $0$); \ie they sign exactly one transfer order per sequence number;
\item They broadcast the new transfer order $O$ to enough authorities until they (eventually) obtain a certificate $C$;
\item They successfully broadcast the certificate $C$ to a quorum of authorities.
\end{enumerate}

\para{\sysname Client} To address the correctness rules above, our reference implementation of a \sysname client holds and persists the following minimal state:
\begin{itemize}
\setlength\itemsep{0em}
\item The address $x$ and the secret key of the account;
\item The \sysname committee;
\item The sequence number to be used in the next transfer;
\item The transfer order that it signed last, in case it is still pending.
\end{itemize}
In this setting, the available balance of a user account is not tracked explicitly but rather evaluated (conservatively) from the \sysmain transactions and the available logs for incoming transfers and outgoing transfers (\Cref{sec:audits}). Evaluating the balance before starting a transfer is recommended, as signing a transfer order with an excessive amount will block (correct) client implementations from initiating further transfers until the desired amount is available.

%% file: sections/security.tex
\section{Security Analysis} \label{sec:security}

Let $\sigma$ denote the current state of the \sysmain. We define $\funding^x(\sigma)$ as the sum of all the amounts transferred to a \sysname address $x$ from the \sysmain:
\[
\funding^x(\sigma) \,=\, \sum_{\left\{\!\!\!\!\scriptsize\begin{array}{l} T \in \transactions(\sigma)\\ \recipient(T)=x \end{array}\right.} \amount(T)
\]
For simplicity, we write $\sum_C \amount(C)$ when we mean to sum over certified transfer orders: $\sum_{O \text { s.t. } \exists C. O = \val(C)} \amount(O)$.

The results presented in this section are proven in \Cref{sec:proofs}. We start with the main safety invariant of \sysname.

\begin{theorem}[Solvency of \sysname] \label{theorem-s6} At any time, the sum of the amounts of all existing certified transfers from \sysname to the \sysmain cannot exceed the funds collected by all transactions on the \sysmain smart contract:
\[
\sum_{\recipient(C) \in \sysmainaddresses} \!\!\amount(C)
\;\leq\; \sum_x \funding^x(\sigma)
\]
\end{theorem}

Next, we describe how receivers of valid transfer certificates can finalize transactions and make funds available on their own accounts (\sysmain and \sysname).

\begin{proposition}[Redeemability of valid transfer certificates to \sysmain] \label{proposition-l1B}
A new valid \sysmain transfer certificate $C$ can always be redeemed by sending a new redeem transaction $T$ to the smart contract.
\end{proposition}

\begin{proposition}[Redeemability of valid transfer certificates to \sysname] \label{proposition-l1A}
Any user can eventually have a valid \sysname transfer certificate $C$ confirmed by any honest authority. \end{proposition}

Specifically, in \Cref{proposition-l1A}, the confirmation order for $C$ is guaranteed to succeed for every honest authority $\alpha$, provided that the user first recovers and transfers to $\alpha$ all the \emph{missing certificates required by $\alpha$}, defined as the sequence $C_k \ldots C_{n-1}$ such that $k = \nextsequencenumber^x(\alpha)$, $x = \sender(C)$, $n = \sequencenumber(C)$, $\sender(C_i) = x$ $(k\leq i \leq n-1)$. The fact that no other certificates need to be confirmed (\eg to credit the balance of $\sender(C)$ itself) is closely related to the possibility of (temporary) negative balances for authorities, and justified by the proof of safety in \Cref{sec:proofs}.

Note that having a \sysname certificate confirmed by an authority $\alpha$ only affects $\alpha$'s recipient and the sender's balances (\ie \emph{redeems the certificate}) the first time it is confirmed.

Finally, we state that \sysname funds credited on an account can always be spent.
We write $\receivedlist(x)$ for the set of \emph{incoming} transfer certificates $C$ such that $\recipient(C)=x$ and $C$ is known to the owner of the account\ $x$.

\begin{proposition}[Availability of transfer certificates] \label{proposition-l2}
Let $x$ be an account owned by a correct user, $n$ be the next available sequence number after the last signed transfer order (if any, otherwise $n=0$), and $O$ be a new transfer order signed by~$x$ with $\sequencenumber(O) = n$ and $\sender(O) = x$.

Assume that the owner of $x$ has secured enough funds for a new order $O$ based on their knowledge of the chain $\sigma$, the history of outgoing transfers, and the set $\receivedlist(x)$. That is, formally:
\[
\begin{array}{l}
\bigg( \amount(O) \; + \; \sum_{\left\{\!\!\!\!\scriptsize\begin{array}{l} \sender(C)=x\\ \sequencenumber(C) < n\end{array}\right.}\!\!\!\!\amount(C) \bigg) \\[0.4em]
\quad \leq \; \bigg( \funding^x(\sigma) \; + \; \sum_{C \in \receivedlist(x)} \amount(C) \bigg)
\end{array}
\]
Then, for any honest authority $\alpha$, the user will always eventually obtain a valid signature of $O$ from $\alpha$ after sending the following orders to $\alpha$:
\begin{enumerate}
\setlength\itemsep{0em}
\item A synchronization order from the \sysmain based on the known state $\sigma$;
\item A confirmation order for every $C \in \receivedlist(x)$, preceded by all the missing certificates required by $\alpha$ (if any) for the sender of $C$;
\item Then, the transfer order $O$.
\end{enumerate}
\end{proposition}

\para{Worst-case efficiency of \sysname clients} To initiate a transfer (\Cref{proposition-l2}) or receive funds (\Cref{proposition-l1A}) from a sender account $x$, a \sysname client must address a quorum of authorities. During the exchange, each authority $\alpha$ may require missing certificates $C_k \ldots C_{n-1}$, where $k = \nextsequencenumber^x(\alpha)$ is provided by $\alpha$. In an attempt to slow down the client, a Byzantine authority could return $k=0$ and/or fail to respond at some point. To address this, a client should query each authority $\alpha$ in parallel. After retrieving the sequence number~$k$, the required missing certificates should be downloaded sequentially, in reverse order, then forwarded to~$\alpha$. Given that \sysname client operations succeed as soon as a quorum of authorities completes their exchanges, this strategy ensures client efficiency despite Byzantine authorities.

%% file: sections/implementation.tex
\section{Implementation} \label{sec:implementation}
We implemented both a \sysname client and a networked multi-core multi-shard \sysname authority in Rust, using Tokio\footnote{\url{https://tokio.rs}} for networking and ed25519-dalek\footnote{\url{https://github.com/dalek-cryptography/ed25519-dalek}} for signatures. For the verification of the multiple signatures composing a certificate we use ed25519 batch verification. To reduce latency we use UDP for \sysname requests and replies, and make the core of \sysname idempotent to tolerate retries in case of packet loss; we also provide an experimental \sysname implementation using exclusively TCP. Currently, data-structures are held in memory rather than persistent storage. 

We implement an authority shard as a separate operating system process with its own networking and Tokio reactor core, to validate the low overhead of intra-shard coordination (through message passing rather than shared memory). We experimented with manually pinning processes to physical cores without a noticeable increase in performance through the Linux \emph{taskset} feature. It seems the Linux OS does a good job in distributing processes and keeping them on inactive cores. We also experimented with a single process multi-threaded implementation of \sysname, using a single Tokio reactor for all shards on multi-core machines. However, this led to significantly lower performance, and therefore we opted for using separate processes even on a single machine for each shard. The exact bottleneck justifying this lower performance---whether at the level of Tokio multi-threading or OS resource management---still eludes us.

The implementation of both server and client is less than 4,000 LOC (of which half are for the networking), and a further 1,375 LOC of unit tests. It required about 2.5 months of work for 3 engineers, and a bit over 1,500 git commits. Keeping the core small required constant re-factoring and its simplicity is a significant advantage of the proposed \sysname design. 
We are open sourcing the Rust implementation, Amazon web services orchestration scripts, benchmarking scripts, and measurements data to enable reproducible results\footnote{
\ifdefined\cameraReady
\url{https://github.com/novifinancial/fastpay}
\else
Link omitted for blind review.
\fi
}.

%% file: sections/evaluation.tex
\section{Evaluation} \label{sec:evaluation}
We evaluate the throughput and latency of our implementation of \sysname through experiments on AWS. We  particularly aim to demonstrate that \first sharding is effective, in that it increases throughput linearly as expected; \second latency is not overly affected by the number of authorities or shards, and remains near-constant, even when some authorities fail; and \third that the system is robust under extremely high concurrency and transaction loads.

\subsection{Microbenchmarks}
We report on microbenchmarks of the single-CPU core time required to process transfer orders, authority signed partial certificates, and certificates. \Cref{table:microbenchark} displays the cost of each operation in micro seconds ($\mu s$) assuming 10 authorities (recall $1\mu s = 10^{-6}s$); each measurement is the result of 500 runs on an Apple laptop (MacBook Pro) with a 2.9 GHz Intel Core i9 (6 physical and 12 logical cores), and 32 GB 2400 MHz DDR4 RAM. The first 3 rows respectively indicate the time to create and serialize \first a transfer order, \second a partial certificate signed by a single authority, and \third a transfer certificate as part of a confirmation order. The last 3 rows indicate the time to deserialize them and check their validity.
The dominant CPU cost involves the deserialization and signature check on certificates ($236\mu s$), which includes the batch verification of the 8 signatures (7 from authorities and 1 from sender). However, deserializing orders ($58\mu s$) and votes ($60\mu s$) is also expensive: it involves 1 signature verification (no batching) and creating 1 signature. These results indicate that a single core shard implementation may only settle just over 4,000 transactions per second---highlighting the importance of sharding to achieve high-throughput.

In terms of networking costs, a transfer order is 146 bytes, and the signed response is 293 bytes. This could be reduced by only responding with a signature (64 bytes) rather than the full signed order, but we chose to echo back the order to simplify client implementations. A full certificate for an order is 819 bytes, and the response---consisting of an update on the state of the \sysname account---is 51 bytes. For deployments using many authorities we can compress certificates by using an aggregate signature scheme (such as BLS~\cite{bls}). However, verification CPU costs of BLS only make this competitive for committees larger than 50-100 authorities. We note that all \sysname message types fit within the common maximum transmission unit of commodity IP networks, allowing requests and replies to be executed using a single UDP packet (assuming no packets loss and 10 authorities).

\begin{table}[t]
\centering
\footnotesize
\begin{tabular}{lrr} \toprule
Measure & Mean ($\mu s$) & Std. ($\mu s$) \\
\midrule
Create \& Serialize Order & 27 & 1 \\
Create \& Serialize Partial Cert. & 27 & 2 \\
Create \& Serialize Certificate & 4 & 0 \\
Deserialize \& Check Order & 58 & 1 \\
Deserialize \& Check Partial Cert. & 60 & 1 \\
Deserialize \& Check Certificate & 236 & 10 \\
\bottomrule
\vspace{0.5mm}
\end{tabular}
\caption{\footnotesize Microbenchmark of single core CPU costs of \sysname operations; average and standard deviation of 500 measurements for 10 authorities.}
\label{table:microbenchark}
\end{table}
%

\subsection{Throughput} \label{sec:throughput-eval}

We deploy a \sysname multi-shard authority on Amazon Web Services (Stockholm, eu-north-1 zone), on a m5d.metal instance. This class of instance guarantees 96 virtual CPUs (48 physical cores), on a 2.5 GHz, Intel Xeon Platinum 8175, and 384 GB memory. The operating system is Linux Ubuntu server 18.04, where we increase the network buffer to about 96MB. 
In all graphs, each measurement is the average of 9 runs, and the error bars represent one standard deviation; all experiments use our UDP implementation.
We measure the variation of throughput with the number of shards. Our baseline experiment parameters are: 4 authorities (for confirmation orders), a load of 1M transactions, and applying back-pressure to allow a maximum of 1000 concurrent transactions at the time into the system (\ie the \emph{in-flight} parameter). We then vary these baseline parameters through our experiments to illustrate their impact on performance.
We select 4 authorities as baseline for our experiments to make it easier to compare with other systems' evaluations~\cite{han2019performance}.

\para{Robustness and performance under high concurrency}
Figures~\ref{fig:x-1000000-z-4-transfer} and~\ref{fig:x-1000000-z-4-confirmation} respectively show the variation of the throughput of processing transfer and confirmation orders as we increase the number of shards per authority, from 15 to 85. We measure these by processing 1M transactions, across 4 authorities. \Cref{fig:x-1000000-z-4-transfer} shows that the throughout of transfer orders slowly increases with the number of shards. The in-flight parameter---the maximum number of transactions that is allowed into the system at any time---influences the throughput by about 10\%, and setting it to 1,000 seems optimal for performance. The degree of concurrency in a system depends on the number of concurrent client requests, and we observe that \sysname is stable and performant even under extremely high concurrency peaks of 50,000 concurrent requests. Afterwards, the Operating System UDP network buffers fill up, and the authority network stacks simply drop the requests.

\Cref{fig:x-1000000-z-4-confirmation} shows that the throughput of confirmation orders initially increases linearly with the number of shards, and then reaches a plateau at around 48 shards. This happens because our experiments are run on machines with 48 physical cores, running at full speed, and 48 logical cores.
The in-flight parameter of concurrent requests does not influence the throughput much, but setting it too low (\eg at 100) does not saturate our CPUs. These figures show that \sysname can support up to 160,000 transactions per second on 48 shards (about 7x the peak transaction rate of the Visa payments network~\cite{visa-performance}) while running on commodity computers that cost less than 4,000 USD/month per authority\footnote{AWS reports a price of 5.424 USD/hour for their \texttt{m5d.metal} instances. \url{https://aws.amazon.com/ec2/pricing/on-demand} (January 2020)}.

\begin{figure}[t]
\includegraphics[width=.47\textwidth]{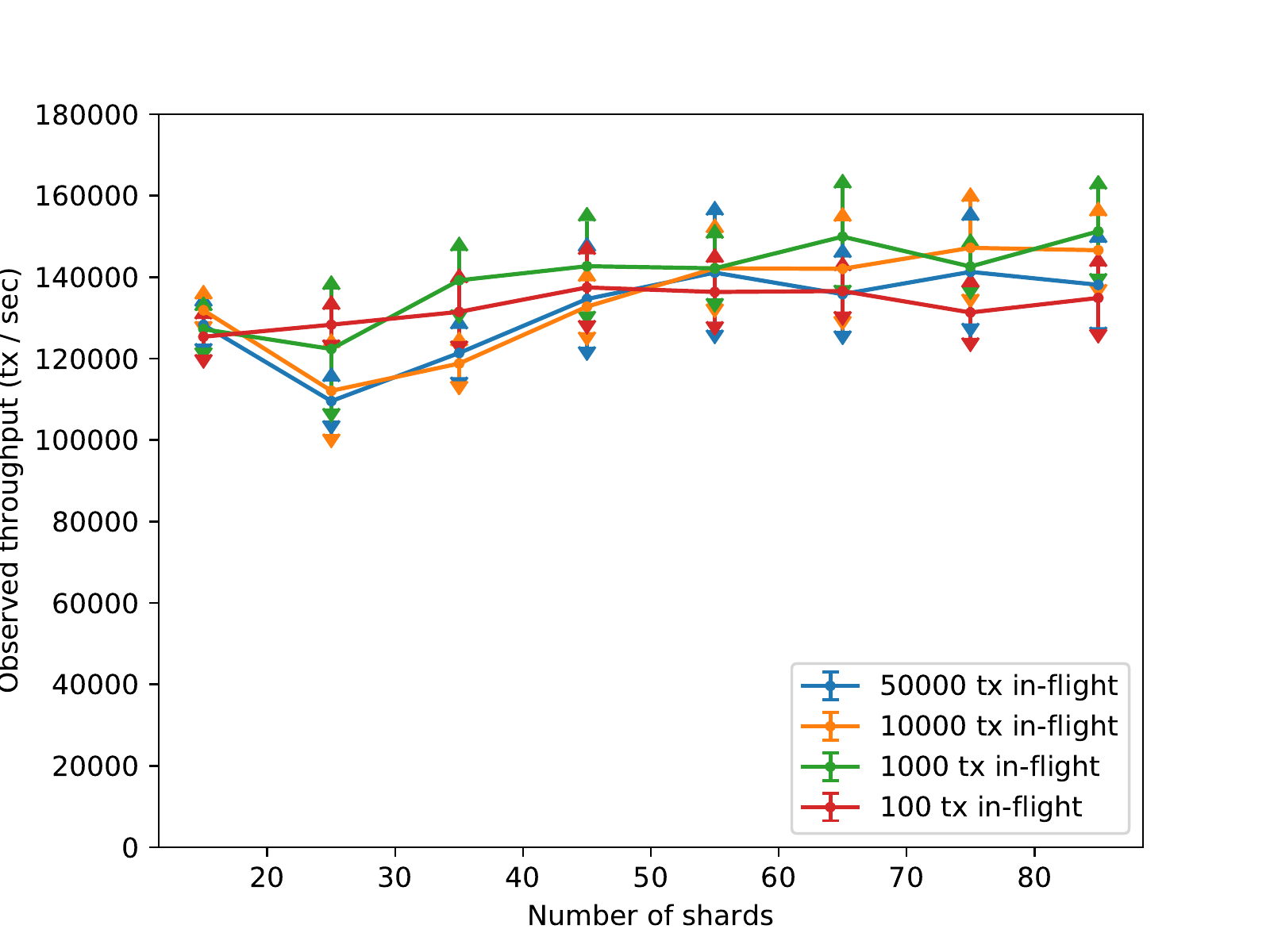}
\caption{\footnotesize Variation of the throughput of transfer orders with the number of shards, for various levels of concurrency (in-flight parameter). The measurements are run under a total load of 1M transactions.}
\label{fig:x-1000000-z-4-transfer}
\end{figure}
\begin{figure}[t]
\includegraphics[width=.47\textwidth]{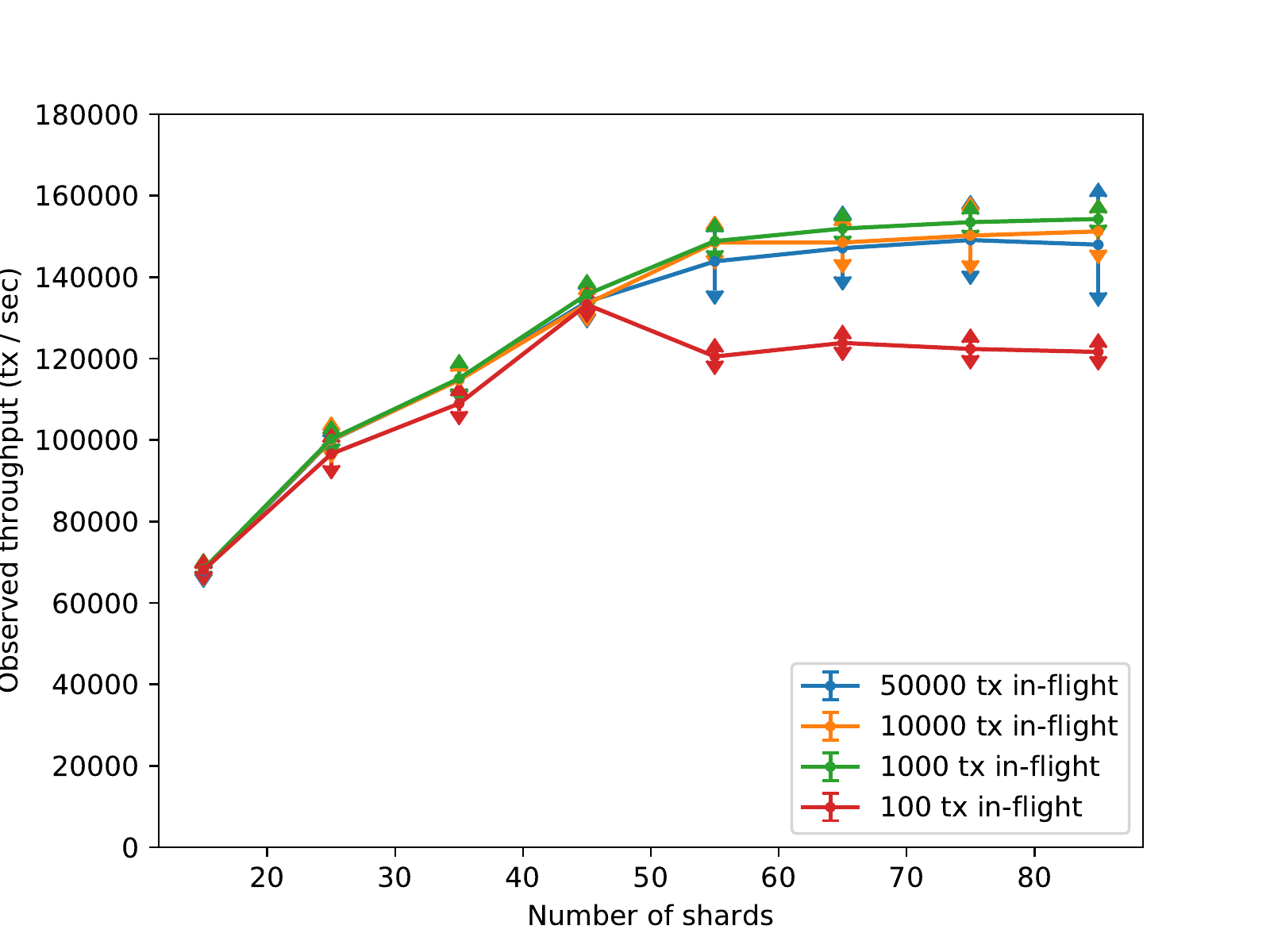}
\caption{\footnotesize Variation of the throughput of confirmation orders with the number of shards, for various levels of concurrency (in-flight parameter). The certificates are issued by 4 authorities, and the measurements are run under a total load of 1M transactions.}
\label{fig:x-1000000-z-4-confirmation}
\end{figure}

\para{Robustness and performance under total system load}
Figures \ref{fig:x-z-1000-4-transfer} and \ref{fig:x-z-1000-4-confirmation} (see \Cref{sec:additional-figures}) show the variation of the throughput of transfer and confirmation orders with the number of shards, for various total system loads---namely the total number of transactions in the test, submitted at the same time.
The goal of this experiment is to analyze the system's performance when experiencing high peaks of utilization. Our results show that the throughput is not affected by the system load.
The tests were performed with 4 authorities, and the client concurrency in-flight parameter set to 1,000.
These figures illustrate that \sysname can process about 160,000 transactions per second even under a total load of 1.5M transactions, and that the total load does not significantly affect performance. These supplement figures \ref{fig:x-1000000-z-4-transfer} and~\ref{fig:x-1000000-z-4-confirmation} that illustrate the concurrent transaction rate (in-flight parameter) also does not influence performance significantly (except when it is too low by under-utilizing the system).

For comparison, in the key experimental work~\cite{han2019performance}, Han~\etal study a number of permissioned systems under a high load. They show that for all of Hyperledger Fabric (v0.6 with PBFT)~\cite{hyperledger-v06}, Hyperledger Fabric (v1.0 with BFT-Smart)~\cite{hyperledger-v10}, Ripple~\cite{ripple-documentation}, and R3 Corda v3.2~\cite{corda-v32}, the successful requests per second \emph{drops to zero} when the transaction rate increases to more than a few thousands transactions per second (notably for Corda only a few hundred). An important exception is Tendermint~\cite{tendermint}, which maintains a processed transaction rate of about 4,000 to 6,000 transactions per second at a high concurrency rate. These findings were confirmed for Hyperledger Fabric that reportedly starts saturating at a rate of 10,000 transactions per second~\cite{nasir2018performance}. In contrast, our results suggest that \sysname stays performant under extremely high rates of concurrent transactions (in-flight parameter) and high work load (total number of transactions processed).

\para{Influence of the number of authorities}
As discussed in \Cref{sec:design}, we expect that increasing the number of authorities only impacts the throughput of confirmation orders (that need to transfer and check transfer certificates signed by $2f+1$ authorities), and not the throughput of transfer orders. \Cref{fig:z-1000000-1000-x-confirmation} confirms that the the throughput of confirmation orders decreases as the number of authorities increases. \sysname can still process about 80,000 transactions per second with 20 authorities (for 75 shards). The measurements are taken with an in-flight concurrency parameter set to 1,000, and under a load of 1M total transactions. We note that for higher number of authorities, using an aggregate signature scheme (\eg BLS~\cite{bls}) would be preferable since it would result in constant time verification and near-constant size certificates. However, since we use batch verification of signatures, the break even point may be after 100 authorities in terms of verification time.

\begin{figure}[t]
\includegraphics[width=.47\textwidth]{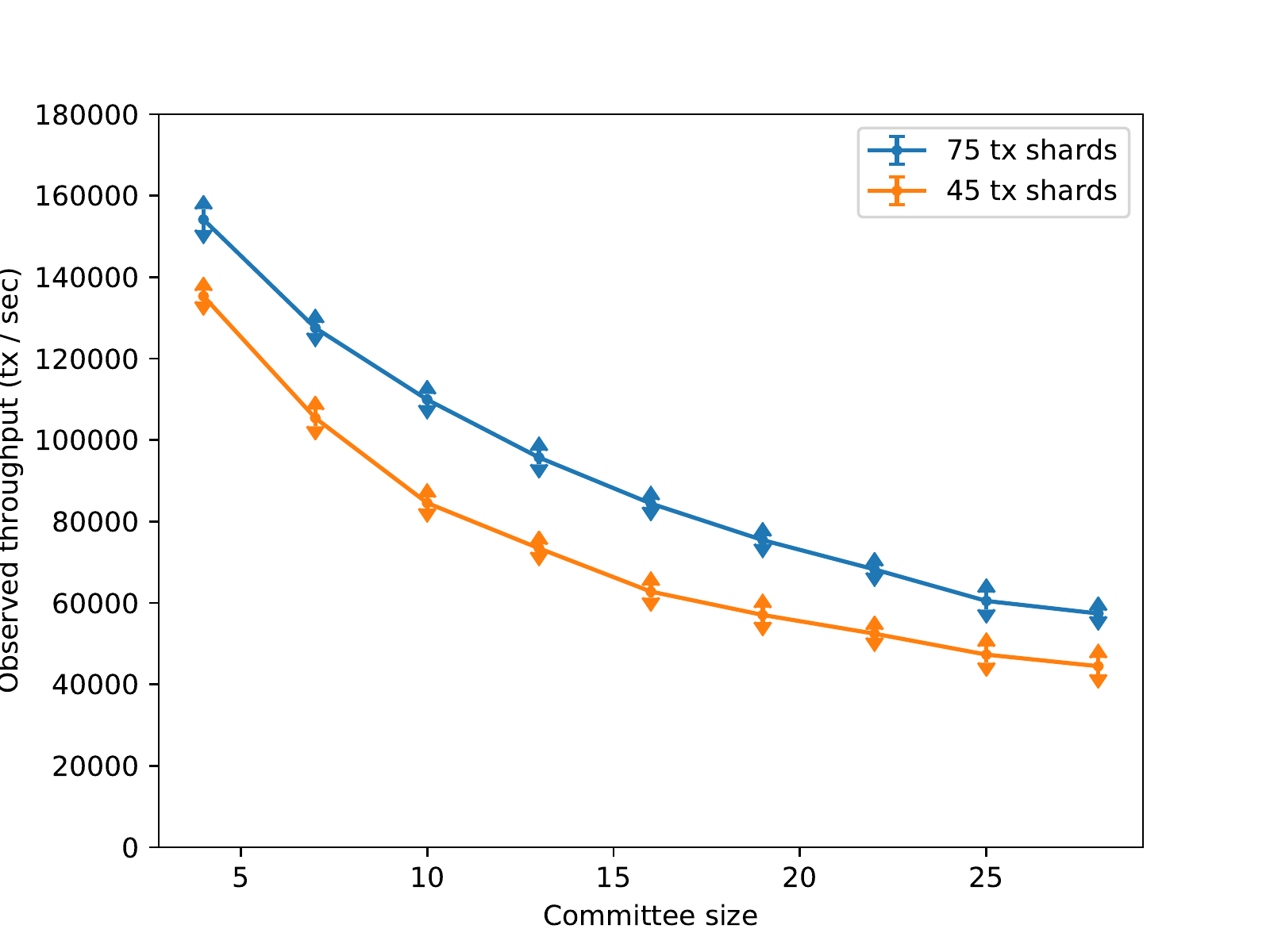}
\caption{\footnotesize Variation of the throughput of confirmation orders with the number of authorities, for various number of shards. The in-flight parameter is set to 1,000 and the system load is of 1M transactions.}
\label{fig:z-1000000-1000-x-confirmation}
\end{figure}
%

\subsection{Latency}
We measure the variation of the client-perceived latency with the number of authorities. We deploy several \sysname multi-shard authorities on Amazon Web Services (all in Stockholm, eu-north-1 zone), each on a m5d.8xlarge instance. This class of instance guarantees 10Gbit network capacity, on a 3.1 GHz, Intel Xeon Platinum 8175 with 32 cores, and 128 GB memory. The operating system is Linux Ubuntu server 16.04. Each instance is configured to run 15 shards.
The client is run on an Apple laptop (MacBook Pro) with a 2.9 GHz Intel Core i9 (6 physical and 12 logical cores), and 32 GB 2400 MHz DDR4 RAM; and connected to a reliable WIFI network. We run experiments with the client in two different locations; \first in the U.K. (geographically close to the authorities, same continent), and \second in the U.S. West Coast (geographically far from the authorities, different continent).
Each measurement is the average of 300 runs, and the error bars represent one standard deviation; all experiments use our UDP implementation.

We observe that the client-authority WAN latency is low for both transfer and confirmation orders; the latency is under 200ms when the client is in the U.S. West Coast, and about 50ms when the client is in the U.K.
\Cref{fig:latency-transfer} illustrates the latency between a client creating and sending a transfer order to all authorities, and receiving sufficient signatures to form a transfer certificate (in our experiment we wait for all authorities to reply to measure the worse case where $f$ authorities are Byzantine). The latency is virtually constant as we increase the number of authorities, due to the client emitting orders asynchronously to all authorities and waiting for responses in parallel. 

\Cref{fig:latency-confirmation} illustrates the latency to submit a confirmation order, and wait for all authorities to respond with a success message. It shows latency is virtually constant when increasing the number of authorities. This indicates that the latency is largely dominated by the network (and not by the verification of certificates). However, since even for 10 authorities a \sysname message fits within a network MTU, the variation is very small. Due to our choice of using UDP as a transport there is no connection initiation delay (as for TCP), but we may observe packet loss under very high congestion conditions. Authority commands are idempotent to allow clients to re-transmit to overcome loss without sacrificing safety.

\begin{figure}[t]
\includegraphics[width=.47\textwidth]{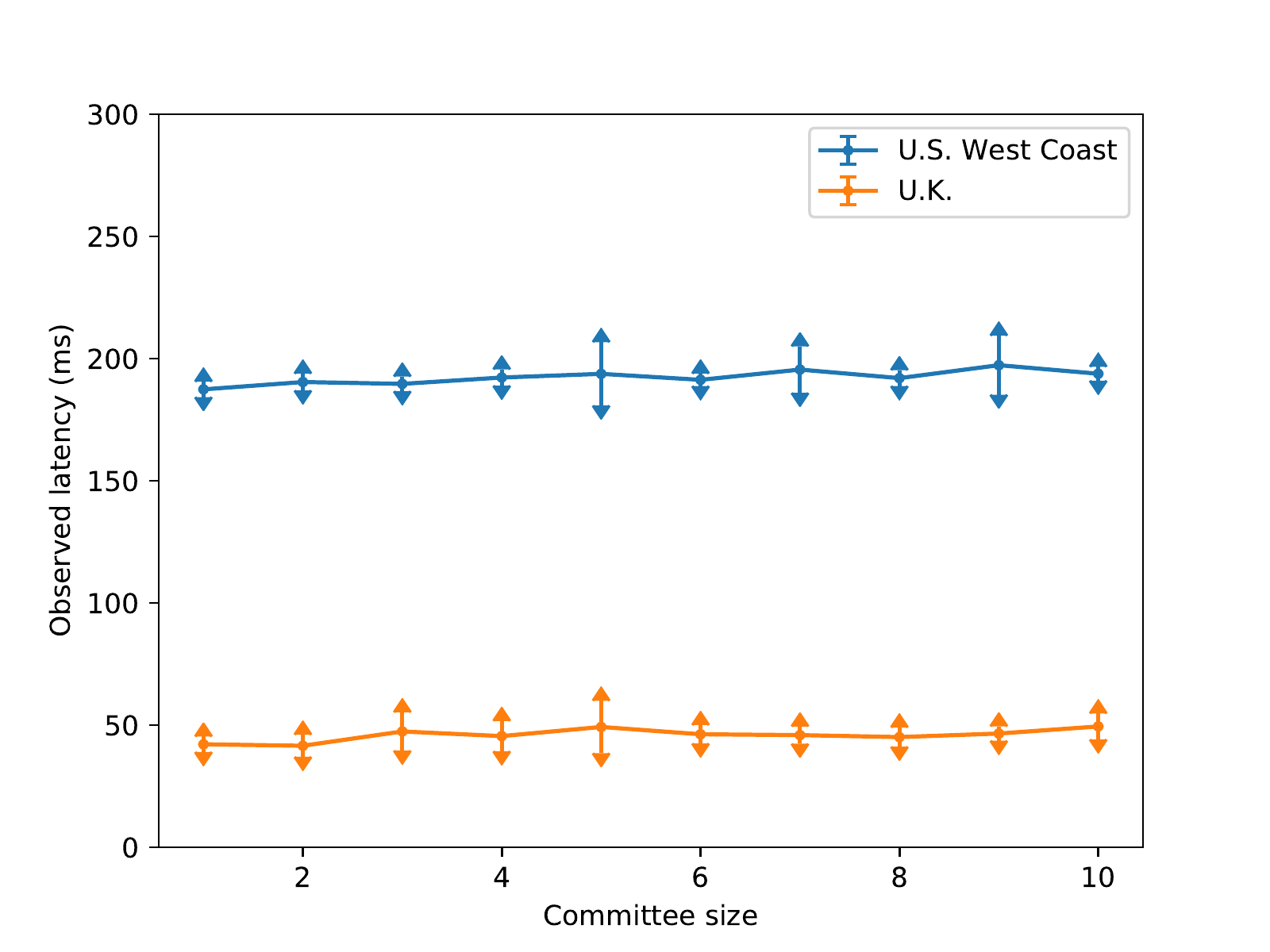}
\caption{\footnotesize Variation of the latency of transfer orders with the number of authorities, for various locations of the client.}
\label{fig:latency-transfer}
\end{figure}
\begin{figure}[t]
\includegraphics[width=.47\textwidth]{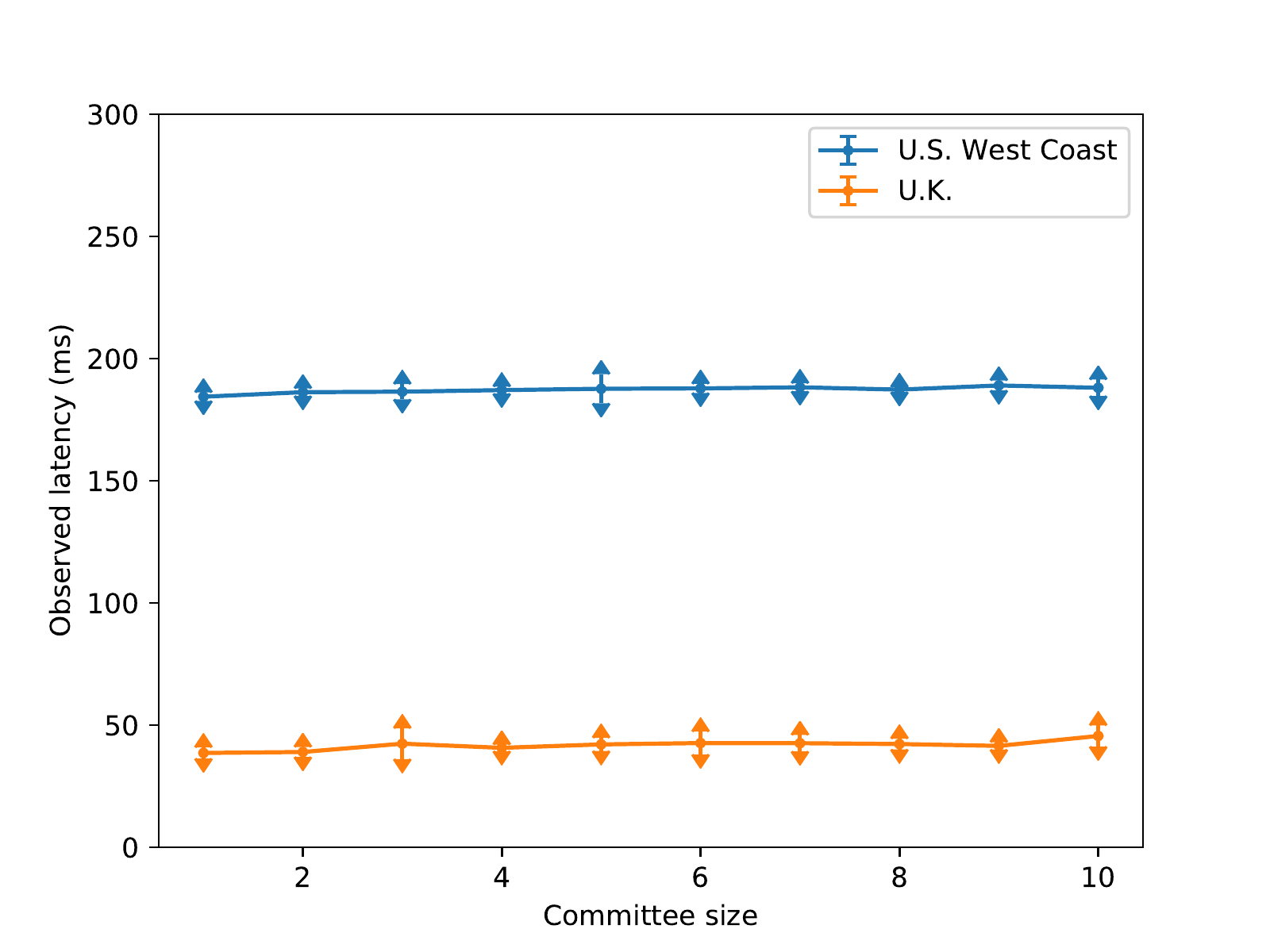}
\caption{\footnotesize Variation of the latency of confirmation orders with the number of authorities, for various locations of the client.}
\label{fig:latency-confirmation}
\end{figure}

\para{Performance under failures} Research literature suggests permissioned blockchains based on (often leader-based) consensus suffer an enormous performance drop when some authorities fail~\cite{DBLP:conf/icdcs/LeeSHKN14}. We measure the effect of authority failure in \sysname and show that latency is not affected when $f$ or fewer authorities are unavailable.

We run our baseline experimental setup (10 authorities distributed over 10 different AWS instances), when a different number of authorities are not available for $f = 0\ldots3$. We measure the 
\begin{wraptable}{r}{0.45\columnwidth}
\centering
\footnotesize
\begin{tabular}{lr}  
\toprule
$f$    & Latency\\
& (ms $\pm$ std) \\
\midrule
0 & $43 \pm 2$\\
1 & $41 \pm 3$\\
2 & $44 \pm 4$\\
3 & $47 \pm 2$\\
\bottomrule
\end{tabular}
\caption{\footnotesize Crash-failure Latency.}
\label{tab:fail}
\end{wraptable}
latency experienced by a client  on the same continent (Europe), sending a transfer order until it forms a valid transfer certificate. \Cref{tab:fail} summarizes the mean latency and standard deviation for different $f$. 
There is no statistically significant difference in latency, no matter how many tolerable failures \sysname experiences (up to $f \leq 3$ for 10 authorities). We also experimented with killing authorities one by one with similar results, up to $f > 3$ when the system did observably lose liveness as expected. The underlying reason for the steady performance under failures is \sysname's lack of reliance on a leader to drive the protocol.

%% file: sections/discussion.tex
\section{Limitations and Future Work} \label{sec:discussion}

\para{Threats to validity of experiments} 
Our experiments represent the best case performance, for a set number of authorities and shards, as they are performed in laboratory conditions. In particular, real-wold transactions may have the same sender account, which would prevent them from being executed in parallel. Further, the throughput evaluation places transaction load on an authority through the local network interface, and therefore does not take fully into account the operating system networking costs of a full WAN stack. Further, our WAN latency experiments were performed against authorities with very low-load. Finally, the costs of persisting databases to storage are not taken into account when measuring latency and throughput (we leave the implementation of low-latency persistent storage to future work).

\para{Integrating privacy} 
As presented, \sysname exposes information about all transactions, namely the sender-recipient accounts and the amounts transferred, as well as the timings of those transfers. Fully integrating stronger privacy protections is a separate research project. However, we want to highlight that the architecture of \sysname is highly compatible with threshold issuance selective disclosure credential designs, such as Coconut~\cite{coconut}. In those schemes a threshold of authorities can jointly sign a credential that the user can subsequently randomize and present to execute a payment. 
Implementing hidden balances, like MinbleWinble~\cite{mimblewimble} and combining them with credentials should be possible---but beyond the scope of the present work.

\para{Checkpointing, authority, and key rotation}
The important enabler for the good performance of \sysname, but also an important limitation, is the fact that authorities do not need to reach consensus on the state of their databases. We demonstrate that payments are secure in this context, but various system maintenance operations are harder to implement. For example, checkpointing the state of all accounts in the systems, to compress the list of stored certificates would be beneficial, but cannot be straightforwardly implemented without consensus. Similarly, it would be beneficial for authorities to be able to rotate in and out of the committee, as well as to update their cryptographic signature keys. Due to the lack of tight synchronization between authorities there is no natural point that guarantees they all update their committees at the same logical time. Further, our proofs of liveness under asynchrony presume that transfer orders and certificates that were once valid, will always be valid. Integrating such governance features into \sysname will require careful design to safely leverage either some timing (synchrony) assumptions or use a more capable (but maybe lower performance) consensus layer, such as one facilitated by the \sysmain.

\para{Economics and fees} Some cost to insert transactions into a system (like fees in Bitcoin), allows for sound accounting and prevents Denial of Service attacks by clients over-using an open system. The horizontal scalability of \sysname alleviates somehow the need to integrate such a scheme, since issues of capacity can be resolved by increasing its capacity through more shards (as well as deploying network level defenses). However, if there was a need to implement fees for using \sysname we would not recommend using micro-payments associated with each payment like in Bitcoin. We would rather recommend allowing a client to deposit some payment into a service account with all authorities, and then allow them to deduct locally some of this fee for any services rendered (namely any signed transfer order or confirmation order processed). In practical terms, the variable costs of processing transactions in \sysname is low. There is no artificial shortage due to lack of scalability, and a flat periodic fee on either senders or recipients might be sufficient to support operations (rather than a charge per transaction).

%% file: sections/related.tex
\section{Related Works} \label{sec:related}
ABC~\cite{abc} is an asynchronous payment system which can be sharded similarly to \sysname to achieve arbitrary throughput. ABC proposes a relaxed notion of consensus where termination is only guaranteed if the sender of the transaction is honest, and similarly to \sysname dishonest users may lock their own account.
\sysname and ABC share similar features but are designed for different purposes. \sysname can be deployed as a RTGS or a (permissionned) side-infrastructure, and heavily focuses on implementation and evaluation; ABC is a permissionless standalone system providing great details on how to run it as an open system based on proof-of-stake~\cite{sok-consensus}, but provides no implementation or evaluation.

Other systems similar to \sysname are Astro~\cite{astro} and Brick~\cite{brick}, which were both developed concurrently to \sysname. Astro relies on \emph{Byzantine reliable broadcast}~\cite{cachinBook} which adds \emph{totality}~\cite{cachinBook} to Byzantine consistent broadcast. This allows Astro to guarantee availability even to incorrect users (while \sysname only guarantees it for corrects users, see \Cref{sec:properties}) at the cost of one extra broadcast step among the authorities. Astro is designed to be a standalone system and does not natively integrate into a \sysmain infrastructure, and does not offer security proofs.
As \sysname, Brick uses Byzantine consistent broadcast as underlying primitive and positions itself as a payment channel. As such, Brick offers details on how to efficiently open and close channels, provides proofs of fraud in case authorities misbehave, and presents detailed incentive mechanisms to keep authorities honest. However, it does not present an implementation or evaluation (it only presents a quick latency benchmark), and only works with two users. In contrast, \sysname focuses on high performance, robustness, and scalability; it provides a scalable architecture and is specifically designed to handle high transaction volumes, from a high number of users.

We now compare \sysname with traditional payment systems and some relevant crypto-currencies.

\para{Traditional payment systems} In the context of traditional payment systems \sysname is a real-time gross settlement system (RTGS)~\cite{rtgs-1,rtgs-2}---payments are executed in close to real-time, there is no netting between participants, and  the transfer of funds is final upon the full payment protocol terminating. All payments are pre-funded so there is no need to keep track of credit or liquidity, which makes the design vastly simpler.

\sysname, from an assurance and performance perspective is significantly superior to deployed RTGS systems: it \first implements a fully Byzantine fault tolerant architecture (established systems rely on master-slave configurations to only recover from few crash failures), \second has higher throughput (as compared, for example with the TARGET2~\cite{target2} European Central Bank RTGS systems that has a target throughput of 500 tx/sec), and \third has faster finality (as compared to TARGET2 providing finality of a few seconds). Since \sysname allows for fast gross settlement, participants are not exposed to credit risk, as in the case of retail payment systems such as VISA and Mastercard (that use daily netting, and have complex financial arrangements to mitigate credit risk in case of bank default). Furthermore, it does achieve both throughput and latency, comparable to those systems combined---about 80,000 tx/sec at peak times, when adding up the throughput of Visa and Mastercard together~\cite{visa-performance,mastercard-performance}.

On the downside, \sysname lacks certain features of mature RTGS systems: in particular it does not support Delivery-on-Payment transactions that atomically swap securities when payment is provided, or Payment-versus-Payment, that atomically swap amounts in different currencies to minimize the risk of foreign exchange transactions. These require atomic operations across accounts controlled by different users, and would therefore require extending \sysname to support them (namely operations with consensus number of 2 per Herlihy~\cite{herlihy1991wait}), which we leave for future work.

\para{Crypto-currencies} \sysname provides high assurance in the context of Byzantine failures within its infrastructure. In that respect it is comparable with systems encountered in the space of permissioned blockchains and crypto-currencies, as well as their eco-system of payment channels. \sysname is permissioned in that the set of authorities managing the system is closed---in fact we do not even propose a way to rotate those authorities and leave this to future work. Qualitatively, \sysname differs from other permissioned (or permissionless) crypto-currencies in a number of ways: it is secure under full network asynchrony (since it does not require or rely on atomic broadcast channels or consensus, but only consistent broadcast)---leading to higher performance. This direction was explored in the past in relation to central bank crypto-currency systems~\cite{RSCoin} and high performance permissionless systems~\cite{rocket2018snowflake}. It was recently put on a formal footing by Guerraroui~\etal~\cite{consensus-number}. Our work extends this theory to allow increased concurrency, correctness under sharding, and rigorous interfacing with external settlement mechanisms. \sysname achieves auditability through a set of certificates signed by authorities rather than a sequential log of actions (blockchain), which would require authorities to reach agreement on a common sequence. 

Quantitatively, compared with other permissioned systems \sysname is extremely performant. HyperLedger Fabric~\cite{hyperledger} running with 10 nodes achieves about 1,000 transactions per second and a latency of about 10 seconds~\cite{nasir2018performance}; and Libra~\cite{libra} and Corda~\cite{corda, corda-performance} achieve similar performance.
JP Morgan developed a digital coin built from the Ethereum codebase, which can achieve about 1,500 transactions per second with four nodes, and imposing a block time of 1 second~\cite{baliga2018performance}.
Tendermint~\cite{tendermint} reportedly achieves 10,000 transactions per second with 4 nodes, with a few seconds latency~\cite{tendermint-performance}. However, as we discussed in \Cref{sec:evaluation}, many of those systems see their performance degrading dramatically under heavy load---whereas \sysname performs as expected.

\sysname can be used as a side chain of any crypto-currency with reasonable finality guarantees, and sufficient programmability. As compared to bilateral payment channels it is superior in that it allows users to pay anyone in the system without locking liquidity into the bilateral channel, and is fully asynchronous. However, \sysname does rely on an assumption of threshold non-Byzantine authorities for safety and liveness, whereas payment channel designs only rely on network synchrony for safety and liveness (safety may be lost under conditions of asynchrony). As compared to traditional payment channel networks (such as the lighting network~\cite{lightning}) \sysname is simpler and does not require complex path finding algorithms~\cite{lightning, flare, grunspan2018ant, sivaraman2018routing}. 

%% file: sections/conclusion.tex
\section{Conclusion} \label{sec:conclusion}

\sysname is a settlement layer based on consistent broadcast channels, rather than full consensus. The \sysname design leverages the nature of payments to allow for asynchronous payments into accounts, and optional interactions with an external \sysmain to build a practical system, while providing proofs of both safety and liveness; it also proposes and evaluates a design for sharded implementation of authorities to horizontally scale and match any throughput need. 

The performance and robustness of \sysname is beyond and above the state of the art, and validates that moving away from both centralized solutions and full consensus to manage pre-funded retail payments has significant advantages. Authorities can jointly process tens of thousands of transactions per second (we observed a peak of 160,000 tx/sec) using merely commodity hardware and lean software. A payment confirmation latency of less than 200ms between continents make \sysname practical for point of sale payments---where goods and services need to be delivered fast and in person. Pretty much instant settlement enables retail payments to be freed from intermediaries, such as banks payment networks, since they eliminate any credit risk inherent in deferred netted end-of-day payments, that underpin today most national Fast Payment systems~\cite{bolt2014fast}. Further, \sysname can tolerate up to one-third of authorities crashing or even becoming Byzantine without losing either safety or liveness (or performance). This is in sharp contrast with existing centralized settlement layers operating on specialized mainframes with a primary / backup crash fail strategy (and no documented technical strategy to handle Byzantine operators). Surprisingly, it is also in contrast with permissioned blockchains, which have not achieved similar levels of performance and robustness yet, due to the complexity of engineering and scaling full \bftlong consensus protocols.

%% file: sections/proofs.tex
\section{Proofs of Security} \label{sec:proofs}
We now prove the results presented in \Cref{sec:security}.

\subsection{Additional Notations}
We define $\funding^x(\alpha)$ as the sum of all the amounts received from the \sysmain by a \sysname address $x$, as seen at a given time by an authority\ $\alpha$:
\[\funding^x(\alpha) = \sum_{\sync \in \synchronizationlog^x(\alpha)} \amount(\sync)\]

\subsection{Safety}

\begin{lemma}[Transfer certificate uniqueness] \label{lemma-s0}
If
\[
\left\{
\begin{array}{l}
\sender(C) = \sender(C'), \; \text{and} \\
\sequencenumber(C) = \sequencenumber(C')
\end{array}
\right.
\]
then $C$ and $C'$ certify the same transfer order:
\[\val(C) = \val(C') \]
\end{lemma}
\begin{proof}
Both certificates $C$ and $C'$ are signed by a quorum of authorities. By construction, any two quorums intersect on at least one honest authority. Let $\alpha$ be an honest authority in both quorums. $\alpha$ signs at most one transfer order per sequence number, thus $C$ and $C'$ certify the same transfer order.
\end{proof}

\begin{lemma}[\sysname invariant] \label{lemma-s1}
For every honest authority $\alpha$, for every account $x$, it holds that
\[
\begin{array}{l}
\bigg( \accountbalance^x(\alpha) \;+\; \sum_{C \in \confirmedlog^x(\alpha)} \amount(C) \bigg)\\[0.6em]
\quad \leq \bigg( \funding^x(\alpha) \;+\; \sum_{C \in \receivedlist^x(\alpha)} \amount(C) \bigg)
\end{array}
\]
Besides, if $n = \nextsequencenumber^x(\alpha)$, we have that $\confirmedlog^x(\alpha) = \{C_0 \ldots C_{n-1}\}$ for some certificates $C_k$ such that $\sequencenumber(C_k) = k$ and $\sender(C_k)=x$.
\end{lemma}
\begin{proof}
By construction of the \sysname authorities (\Cref{fig:code} and \Cref{fig:code_core}):
Whenever a confirmed certificate $C$ is added to $\confirmedlog^x(\alpha)$, $\accountbalance^x(\alpha)$ is decreased by $\amount(C)$, and the value $\nextsequencenumber^x(\alpha)$ is incremented into $\sequencenumber(C) + 1$.
Any new synchronization order equally increases $\accountbalance^x(\alpha)$ and $\funding^x(\alpha)$.
Whenever a confirmed certificate $C$ is added to $\receivedlist^x(\alpha)$, eventually $\accountbalance^x(\alpha)$ is increased once by $\amount(C)$. (This may take time due to cross-shard updates.)
\end{proof}

\begin{lemma}[\sysmain invariant] \label{lemma-s2}
The total balance of all \sysname accounts on \sysmain is such that
\[
\begin{array}{l}
\totalbalance(\sigma) = \bigg( \sum_x \funding^x(\sigma) - \\
\qquad \sum_{C \in \redeemlog(\sigma)}\amount(C) \bigg)
\end{array}
\]
\end{lemma}
\begin{proof}
By construction of the smart contract handling funding and redeeming transactions (\Cref{sec:libra-to-fastpay}): whenever a funding transaction $T$ is executed by the smart contract, both $\funding^x(\sigma)$ and $\totalbalance(\sigma)$ increase by $\amount(T)$.

Conversely, $\totalbalance(\sigma)$ decreases by $\amount(C)$ whenever a \sysmain transfer certificate $C$ is redeemed on-chain and added to $\redeemlog(\sigma)$.
\end{proof}

\begin{lemma}[Funding log synchronization] \label{lemma-s3}
For every honest authority $\alpha$ and every account $x$, it holds that
\[ \funding^x(\alpha) \leq \funding^x(\sigma) \]
\end{lemma}
\begin{proof}
By definition of the synchronization with the \sysmain (\Cref{sec:libra-to-fastpay}), and by security of the \sysmain and its client, we note that $\funding^x(\alpha)$ only increases after a funding transaction has already increased $\funding^x(\sigma)$ by the same amount.
\end{proof}

\begin{lemma}[Balance check] \label{lemma-s4}
For every honest authority $\alpha$, when an order $O = \pendingconfirmation^x(\alpha) $ is pending, we have \[\amount(O) \leq \accountbalance^x(\alpha)\]
\end{lemma}
\begin{proof}
By construction of the \sysname authorities (\Cref{fig:code}), if $O = \pendingconfirmation^x(\alpha)$, then $O$ was successfully processed by~$\alpha$ as a new transfer order from account $x$. At the time of the request, $\amount(O)$ did not exceed the current balance $B$. Since $O$ is still pending, in the meantime, no other transfer certificates from account $x$ have been confirmed by $\alpha$. (A confirmation would reset the field $\pendingconfirmation$ and prevent $O$ from being pending again due to increasing sequence numbers.) Therefore, the balance did not decrease, and $\accountbalance^x(\alpha) \geq B \geq \amount(O)$.
\end{proof}

\begin{proposition}[Account safety] \label{proposition-s5}
For every account $x$, at any given time, we have that
\[
\begin{array}{l}
\sum_{\sender(C)=x} \amount(C) \;\leq\; \\
\\
\qquad \funding^x(\sigma) \;+\; \sum_{\recipient(C)=x} \amount(C)
\end{array}
\]
\end{proposition}
\begin{proof}
Let $n$ be the highest sequence number of a transfer certificate $C_n$ from $x$.
Let $\alpha$ an honest authority whose signature is included in the certificate. At the time of the signature,  we had $\val(C_n) = \pendingconfirmation^x(\alpha)$. Therefore, by \Cref{lemma-s1} and \Cref{lemma-s4}, we have
\[
\begin{array}{l}
\bigg( \amount(C_n) \;+\; \sum_{C\in\confirmedlog^x(\alpha)} \amount(C) \bigg) \\[0.6em]
\quad \leq \bigg( \funding^x(\alpha) \;+\; \sum_{C \in \receivedlist^x(\alpha)} \amount(C) \bigg)
\end{array}
\]
Given that $n$ is the highest sequence number, by \Cref{lemma-s0} and \Cref{lemma-s1}, the left-hand term exactly covers the certified transfer orders from $x$ and is equal to $\sum_{\sender(C)=x} \amount(C)$.

Given that amounts are non-negative, for every honest node $\alpha$, we have \[\sum_{C \in \receivedlist^x(\alpha)} \amount(C) \;\leq\; \sum_{\recipient(C)=x}\amount(C)\]

Finally, $\funding^x(\alpha) \leq \funding^x(\sigma)$ by \Cref{lemma-s3}.
\end{proof}

\begin{proof}[Proof of \Cref{theorem-s6} (Solvency)] By applying \Cref{proposition-s5} on every account and summing, we obtain:
\[
\begin{array}{ll}
\sum_x \funding^x(\sigma) \geq \\[0.4em]
\qquad \bigg( \sum_{C} \amount(C) - \sum_{\recipient(C)\in \sysnameaddresses} \amount(C) \bigg) \\[0.6em]
\qquad = \; \sum_{\recipient(C) \in \sysmainaddresses} \amount(C)
\end{array}
\]
\end{proof}


\subsection{Liveness}

\begin{proof}[Proof of \Cref{proposition-l1B} (Redeeming to \sysmain)]
We have seen in \Cref{theorem-s6} that the smart contact always has enough funding for all certified \sysmain transfer orders. The definition of $\redeemlog(\sigma)$ (\Cref{sec:libra-to-fastpay}) thus ensures that any new certified \sysmain transfer order can be redeemed exactly once.
\end{proof}

\begin{proof}[Proof of \Cref{proposition-l1A} (Redeeming to \sysname).]
If a certificate $C$ exists for account $x$ and sequence number $n$, this means at least $f+1$ honest authorities contributed signatures to the transfer order $O = \val(C)$.
By construction of \sysname, these authorities have received (\Cref{fig:code}) and will keep available (\Cref{sec:audits}) all the previous confirmation orders $C_0 \ldots C_{n-1}$ with $\sender(C_k) = x$, $\sequencenumber(C_k) = k$. Therefore, any client can retrieve them and eventually bring any other honest authority up to date with $C$.
\end{proof}

\begin{proof}[Proof of \Cref{proposition-l2} (Availability of certificates)]
Let $B \geq \amount(O)$ be the value defined as follows at the time of the creation of the new transfer order $O$:
\[
\begin{array}{l}
B \; = \; \bigg( \funding^x(\sigma) \;-\; \sum_{\left\{\!\!\!\!\scriptsize\begin{array}{l} \sender(C)=x\\ \sequencenumber(C) < n\end{array}\right.}\!\!\!\!\amount(C) \\[0.4em]
\quad \; + \; \sum_{C \in \receivedlist(x)} \amount(C) \bigg)
\end{array}
\]

By a case analysis similar to the proof of \Cref{lemma-s1}, provided that the owner of $x$ is communicating the information described in \Cref{proposition-l2} to the authority $\alpha$, it will hold eventually that $\accountbalance^x(\alpha) \geq B \geq \amount(O)$ and $\nextsequencenumber^x(\alpha) = n$. We deduce that eventually $\alpha$ will accept the transfer order~$O$ and make the value of its signed (pending) order available.
\end{proof}

\subsection{Performance under Byzantine Failures}

The \sysname protocol does not rely on any designated leader (like PBFT~\cite{castro1999practical}) to make progress or create proposals; \sysname authorities do not directly communicate with each other, and their actions are symmetric. Clients create certificates by gathering the first $2f+1$ responses to a valid transfer order, and no action of a Byzantine authority may delay the creation of a certificate. A Byzantine authority may not even present a signature on a different order as a response to confuse a correct client, since it would have to be signed by the correct payer. Subsequently, a correct client submits the confirmation order to all authorities. Again, Byzantine authorities cannot in any way delay honest authorities from processing the payment locally in their databases, and enabling a subsequent payment for the sending account.

Byzantine clients may attempt denial of service attacks by over-using the system, and for example creating a very large number of receiving accounts (this could be disincentivized by charging some fee for an account creation). However, an attempt to equivocate by sending two transfer orders for a single sequence number could either result in their own account being locked (no single transfer order can achieve $2f+1$ signatures to form a certificate and move to the next sequence number), or one of them succeeding---neither of which degrade performance. Transfer orders with insufficient funds or incorrect sequence numbers are simply rejected, which does not significantly affect performance (if anything they do not result in confirmation orders that are more costly to process than transfer orders, see \Cref{sec:evaluation}).

%% file: sections/code-listings.tex
\section{Code Listings for Core Operations} \label{sec:code-listings}

Algorithms for the core authority operations are simplified directly from the Rust implementation. We omit explicit typing, details of error messages returned, de-referencing, and managing variable ownership. The macro \emph{ensure}, returns with an error unless the condition is fulfilled, and \emph{bail} always returns with an error.

\begin{figure}
\begin{lstlisting}[mathescape=true]
fn handle_cross_shard_commit($\authority$, $\cert$) -> Result {
    let $\transfer$ = $\val(\cert)$;
    let recipient = match $\recipient(\transfer)$ {
        Address::FastPay(recipient) => recipient,
        Address::Primary(_) => { bail!(); }
    };
    ensure!($\authority$.in_shard(recipient));
    let recipient_account = $\accounts(\authority)$.get(recipient)
        .or_insert(AccountState::new());
    recipient_account.balance += $\amount(\transfer)$;
    Ok()
}

fn handle_primary_synchronization_order($\authority$, $\sync$) -> Result {
    /// Update $\recipient(S)$ assuming that $\sync$ comes from
    /// a trusted source (e.g. Primary client).
    let recipient = $\recipient(\sync)$;
    ensure!($\authority$.in_shard(recipient));

    if $\transactionindex(\sync)$ <= $\lasttransactionindex(\authority)$ {
        /// Ignore old synchronization orders.
        return Ok();
    }
    ensure!($\transactionindex(\sync)$ == $\lasttransactionindex(\authority)$ + 1);

    $\lasttransactionindex(\authority)$ += 1;
    let recipient_account = $\accounts(\authority)$.get(recipient)
        .or_insert(AccountState::new());
    recipient_account.balance += $\amount(\sync)$;
    recipient_account.synchronized.push($\sync$);
    Ok()
}
\end{lstlisting}
\vspace{-1em}
\caption{Authority algorithms for cross-shard updates and (\sysmain) synchronization orders.}
\label{fig:code_core}
\end{figure}

%% file: sections/additional-figures.tex
\section{Additional Figures} \label{sec:additional-figures}
Figures \ref{fig:x-z-1000-4-transfer} and \ref{fig:x-z-1000-4-confirmation} show the increase of throughput of transfer and confirmation orders with the number of shards, for various total system loads. They complement \Cref{sec:throughput-eval} by showing that the throughput is not affected by the system load.

\begin{figure}[t]
\includegraphics[width=.47\textwidth]{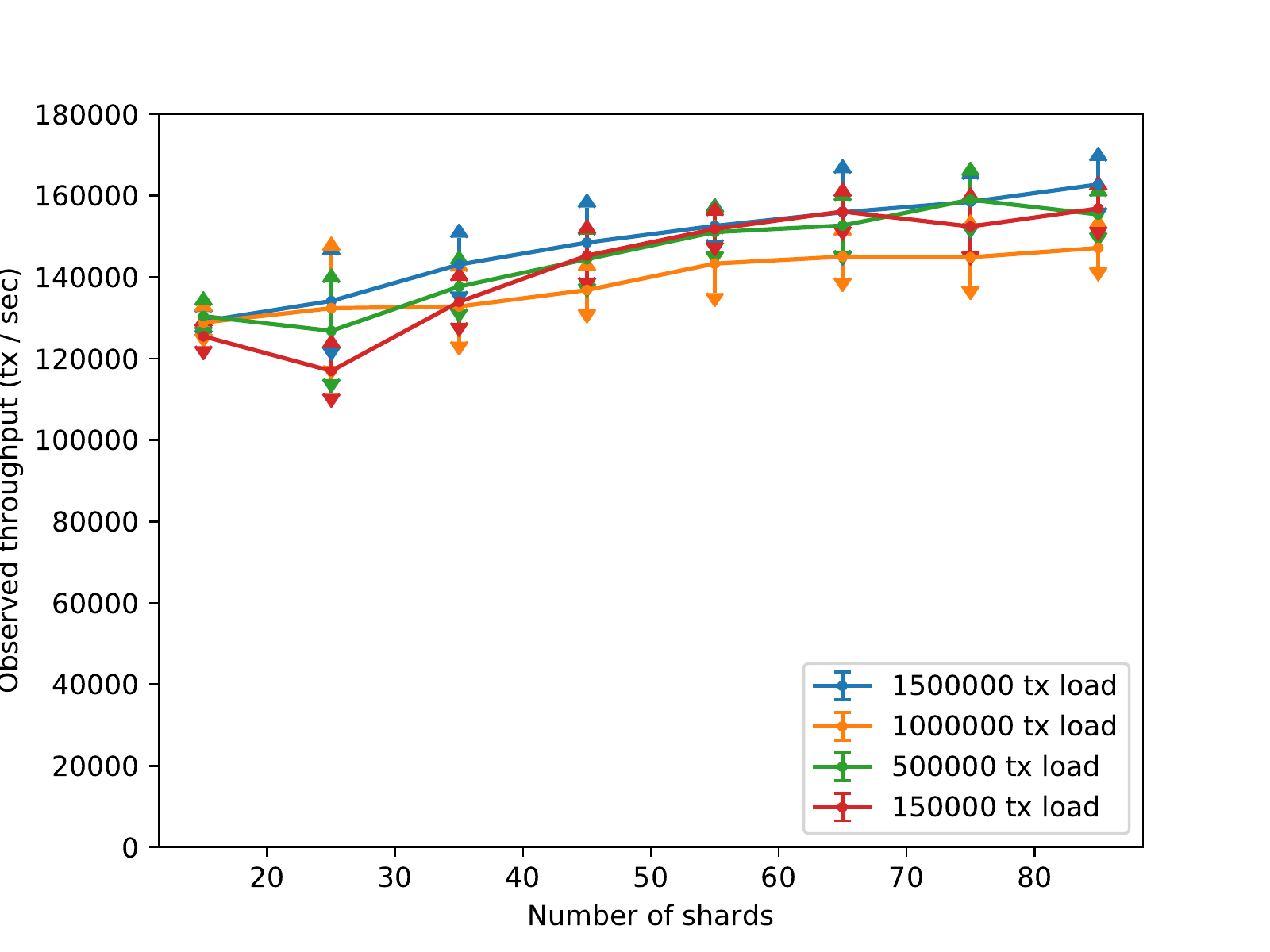}
\caption{Variation of the throughput of transfer orders with the number of shards, for various loads. The in-flight parameter is set to 1,000.}
\label{fig:x-z-1000-4-transfer}
\end{figure}
\begin{figure}[t]
\includegraphics[width=.47\textwidth]{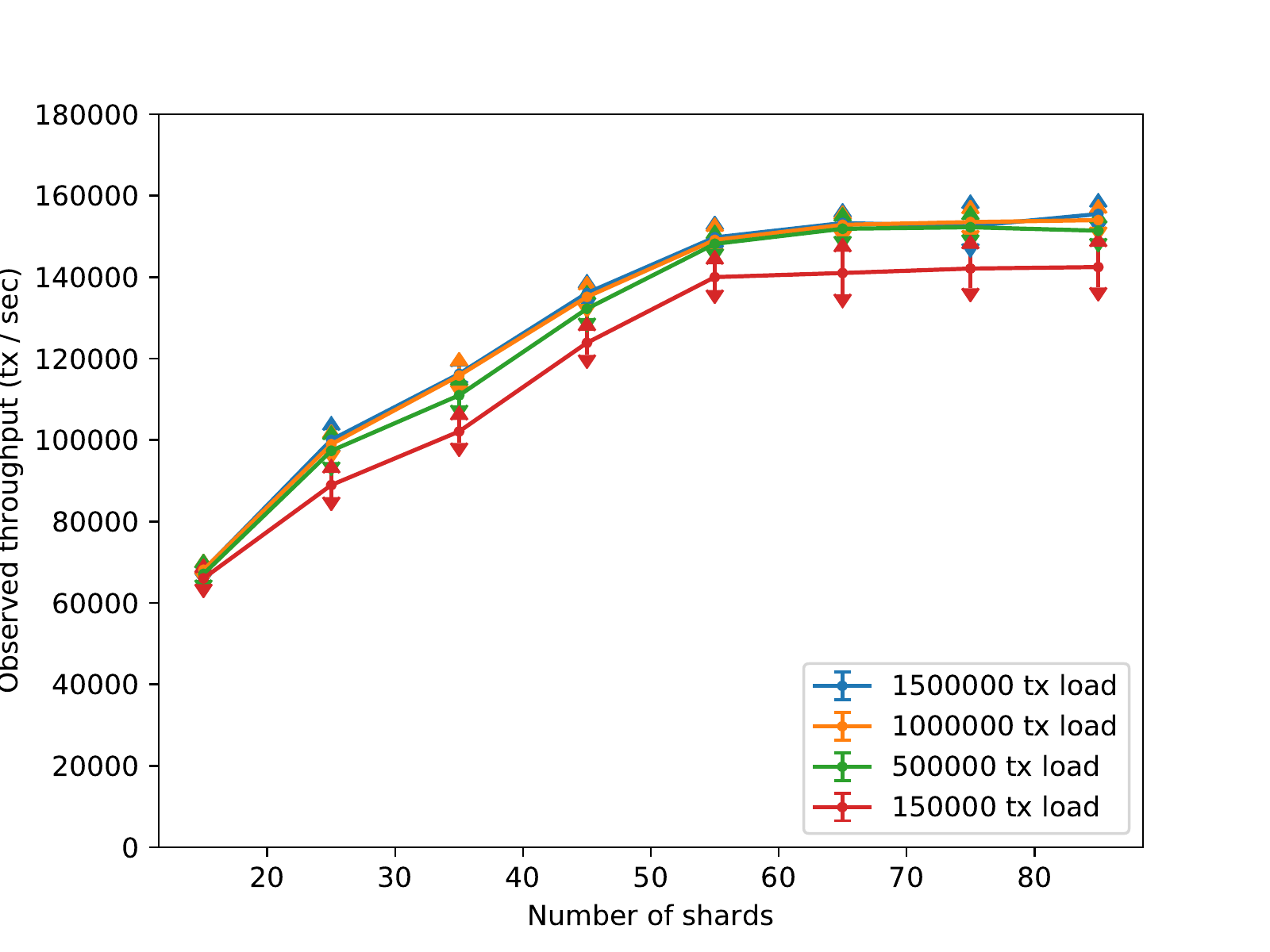}
\caption{Variation of the throughput of confirmation orders with the number of shards, for various loads. The certificates are issued by 4 authorities, and the in-flight parameter is set to 1,000.}
\label{fig:x-z-1000-4-confirmation}
\end{figure}